\documentclass[twocolumn,pra,letterpaper,superscriptaddress]{revtex4-1}
\usepackage{hyperref}
\usepackage{caption,float,latexsym,amsmath,amssymb,amsfonts,graphicx,color,amsthm}
\usepackage{enumerate, subfigure}
\usepackage[dvipsnames]{xcolor}

\newcommand{\cC}{\mathcal{C}}

\newcommand{\cE}{\mathcal{E}}

\newcommand{\cG}{\mathcal{G}}
\newcommand{\cH}{\mathcal{H}}

\newcommand{\cP}{\mathcal{P}}
\newcommand{\cQ}{\mathcal{Q}}
\newcommand{\cR}{\mathcal{R}}

\newcommand{\cT}{\mathcal{T}}

\newcommand{\tr}{\text{tr}}

\newtheorem{theorem}{Theorem}

\begin{document}

\title{Pretty good state transfer via adaptive quantum error correction}
\author{Akshaya Jayashankar}
\affiliation{Department of Physics, Indian Institute of Technology Madras, Chennai, India~600036}
\author{Prabha Mandayam}
\affiliation{Department of Physics, Indian Institute of Technology Madras, Chennai, India~600036}

\begin{abstract}
We examine the role of quantum error correction (QEC) in achieving pretty good quantum state transfer over a class of $1$-d spin Hamiltonians. Recasting the problem of state transfer as one of information transmission over an underlying quantum channel, we identify an adaptive QEC protocol that achieves pretty good state transfer. Using an adaptive recovery and approximate QEC code, we obtain explicit analytical and numerical results for the fidelity of transfer over ideal and disordered $1$-d Heisenberg chains. In the case of a disordered chain, we study the distribution of the transition amplitude, which in turn quantifies the stochastic noise in the underlying quantum channel. Our analysis helps us to suitably modify the QEC protocol so as to ensure pretty good state transfer for small disorder strengths and indicates a threshold beyond which QEC does not help in improving the fidelity of state transfer. 
\end{abstract}

\maketitle
\section{Introduction}

Quantum communication entails transmission of an arbitrary quantum state from one spatial location to another. Spin chains are a natural medium for quantum state transfer over short distances, with the dynamics of the transfer being governed by the Hamiltonian describing the spin-spin interactions along the chains. Starting with the original proposal by Bose~\cite{bose} for state transfer via a $1$-d Heisenberg chain, several protocols have been developed for {\it perfect} as well as {\it pretty good} quantum state transfer via spin chains. 

Perfect state transfer protocols typically involve engineering the coupling strengths between the spins in such a way as to ensure perfect fidelity between the state of sender's spin and that of the receiver's spin~\cite{christandl,christandl2005perfect,albanesemirror,karbach,di}. Alternately, there have been proposals to use multiple spin chains in parallel, and apply appropriate encoding and decoding operations at the sender and receiver's spins so as to transmit the state perfectly~\cite{conclusive,perfect,efficient}. Experimentally, perfect state transfer protocols have been implemented in various architectures including nuclear spins~\cite{bochkin} and photonic lattices using coupled waveguides~\cite{perez2013,chapman}. 

Relaxing the constraint of perfect state transfer, protocols for pretty good transfer aim to identify optimal schemes for transmitting information with high fidelity across permanently coupled spin chains~\cite{godsil2012,godsil}. One approach is for example to encode the information as a Gaussian wave packet in multiple spins at the sender's end~\cite{osborne,hasel}. Moving away from ideal spin chains, quantum state transfer has also been studied over disordered chains, both with random couplings and as well as random external fields~\cite{perfect, ashhab, Chiara}.  

Here, we study the problem of pretty good state transfer from a quantum channel point of view. It is known~\cite{bose} that state transfer over an ideal $XXX$ chain (also called the Heisenberg chain) can be realized as the action of an amplitude damping channel~\cite{Nielsen} on the encoded state. Naturally, this leads to the question of whether  quantum error correction (QEC) can improve the fidelity of quantum state transfer. QEC-based protocols that achieve pretty good transfer have been developed for noisy $XX$~\cite{kay, kay2018perfect} and Heisenberg spin chains~\cite{allcock}. 

In our work we study the role of adaptive QEC in achieving pretty good transfer over a class of $1$-d spin systems which preserve the total spin. This includes both the $XX$ as well as the Heisenberg chains, and more generally, the $XXZ$ chain. We use an {\it approximate} QEC (AQEC) code, which has been shown to achieve the same level of fidelity as perfect QEC codes for certain noise channels while making use of fewer physical resources~\cite{Leung, Fletcher, HKN_PM2010,cafaro}. Our protocol involves the use of multiple identical spin chains in parallel, with the information encoded in an entangled state across the chains. This is in contrast to the protocols in~\cite{kay, kay2018perfect} which use perfect QEC codes and encode into multiple spins on a single chain. Using the worst-case fidelity between the states of the sender and receiver's spins as the figure of merit, we demonstrate that pretty good state transfer maybe achieved over a class of spin-preserving Hamiltonians using an approximate code and a channel-adapted recovery map.

Finally, we present explicit results for the fidelity of state transfer obtained using our QEC scheme, for ideal as well as disordered $XXX$ chains. The presence of disorder in a $1$-d spin chain is known to lead to the phenomenon of localization~\cite{anderson}. Here, we analyze the distribution of the transition amplitude for a disordered $XXX$ chain, with random coupling strengths which are drawn from a uniform distribution. We modify the QEC protocol suitably so as to ensure pretty good transfer when the disorder strength is small. As the disorder strength increases, our analysis points to a threshold beyond which QEC does not help in improving the fidelity of state transfer. 

The rest of the paper is organized as follows. We discuss the basic state transfer protocol over a general class of spin-preserving Hamiltonians and the underlying quantum channel description in Sec.~\ref{sec:prelim}. We discuss the adaptive QEC protocol and the resulting fidelity in Sec.~\ref{sec:ideal}. We present results specific to the ideal $XXX$ chain in Sec.~\ref{sec:Heisenberg} and discuss the disordered chain in Sec.~\ref{sec:disorder}. Finally, we summarize our conclusions in Sec.~\ref{sec:concl}.

\section{Preliminaries}\label{sec:prelim}

We consider a general $1$-d spin chain with nearest neighbour interactions described by the Hamiltonian,
\begin{eqnarray} \label{eq:H_gen}
\cH &=& -\sum_{k} J_{k}\left(\sigma^{k}_{x}\sigma^{k+1}_{x}+\sigma^{k}_{y}\sigma^{k+1}_{y}\right) - \sum_{k}\tilde{J}_{k}\sigma_{z}^{k}\sigma^{k+1}_{z} \nonumber \\
&&  + \sum_{k}B_{k}\sigma_{k}^{z},
 \end{eqnarray}
 where,  $\{J_{k}\}>0$ and $\{\tilde{J}_{k}\}>0$ are site-dependent exchange couplings of a ferromagnetic spin chain, $\{B_{k}\}$ denote the magnetic field strengths at each site, and, $(\sigma^{k}_{x},\sigma^{k}_{y},\sigma^{k}_{z})$ are the Pauli operators at the $k^{\rm th}$ site. The spin sites are numbered as $j = 1,2, \ldots ,N$. We assume that the sender's site is the $s^{\rm th}$ spin and receiver's site is the $r^{\rm th}$ spin. 


We denote the ground state of the spin as $|\textbf{0}\rangle = |000\ldots 0\rangle $. Since we are interested in transmitting a qubit worth of information along the chain, we will work within the subspace spanned by the set of single particle excited states $|\textbf{j}\rangle$, with $|\textbf{j}\rangle$ denoting the state with the $j^{\rm th}$ spin alone flipped to $|1\rangle$. The Hamiltonian in Eq.~\eqref{eq:H_gen} preserves the total number of excitations, that is, $\left[ \cH,\sum_{i=1}^{N} \sigma _{z}^{i} \right] = 0 $ and hence the resulting dynamics is restricted to the $(N+1)$-dimensional subspace spanned by the single particle excited states and the ground state. 

The sender encodes an arbitrary quantum state $|\psi_{\rm{in}}\rangle = a|0\rangle + b|1\rangle $ at the  $s^{ th}$ site, with the coefficients $a$ and $b$ parameterized using a pair of angles $(\theta,\phi)$ as $a$=$\cos(\frac{\theta}{2})$, $b$=$e^{-i \phi} \sin(\frac{\theta}{2})$. The initial state of the spin chain is thus given by,
\begin{equation}
|\Psi(0)\rangle = a |\textbf{0}\rangle + b|\textbf{s}\rangle,  
\end{equation}
where $|\textbf{s}\rangle$ is the state of the spin chain with only the $s^{ th}$ spin is flipped to $|1\rangle$ and all other spins set to $|0\rangle$. Under the action of the Hamiltonian $\cH$ described in Eq.~\eqref{eq:H_gen}, after time $t$, the spin chain evolves to the state (here, and in what follows, we set $\hbar = 1$),
\begin{eqnarray}
|\Psi(t)\rangle &=& e^{-i\cH t}|\Psi(0)\rangle,\nonumber \\
 &=& a |\textbf{0}\rangle + b \sum_{j=1}^{N}\langle\textbf{j}\vert e^{-i\cH t}\vert \textbf{s}\rangle|\textbf{j}\rangle. \nonumber
\end{eqnarray}
Following~\cite{bose}, the state of the receiver's spin at the $r^{ th}$ site after time $t$, denoted as $\rho_{\rm out}(t)$, is obtained by tracing out all the other spins from the state of the full spin chain $\rho(t)= \vert\Psi(t)\rangle\langle\Psi(t)\vert$:
\begin{eqnarray}
&& \rho_{\rm out}(t) = \tr_{1,2,\ldots,r-1,r+1,N-1}\left[\rho(t)\right] \nonumber \\
 &=& \left[ |a|^{2} +|b|^{2}\left(1-|f_{r,s}^N (t)|^{2}\right) \right] |0\rangle\langle 0| + ab^{*} (f_{s,r}^{N}(t))^{*}|0\rangle\langle 1| \nonumber\\
&+&  ba^{*}f_{r,s}^{N}(t)|1\rangle\langle 0| + |b|^{2}|f_{r,s}^N(t)|^{2}|1\rangle\langle 1|, \label{eq:rho_out}
\end{eqnarray}
where, 
\begin{equation}
 f_{r,s}^N(t) = \langle \textbf{r} |e^{(-i \cH t)}|\textbf{s} \rangle \label{eq:trans_amp0}
 \end{equation}
is the {\it transition amplitude}, which gives the probability amplitude for the excitation to transition from the $s^{\rm th}$ site to $r^{\rm th}$ site. The function $f_{r,s}^{N}(t)$ satisfies,
\begin{eqnarray}
\sum_{r=1}^{N}|f_{r,s}^{N}(t)|^{2} &=& 1, \, \forall \; s = 1,2,\ldots, N . \nonumber \\
\sum_{k=1}^{N}f^{N}_{r,k}(t)(f^{N}_{k,s}(t))^{*} &=& \delta_{rs} , \, \forall \; k = 1,2,\ldots, N . \label{eq:trans_amp}
\end{eqnarray}
where $\delta_{rs}$ is the delta function with $\delta_{rs} = 1 $ for $r = s$ and $\delta_{rs} = 0$ for $r\neq s$.

As shown in~\cite{bose}, we thus obtain the reduced state in  Eq.~\eqref{eq:rho_out} at receiver's end as the action of a quantum channel on the input state. Specifically, 
\begin{equation} 
\rho_{\rm{out}}(t) = \cE(\rho_{\rm in}) = \sum_{k}E_{k}\rho_{\rm in}E_{k}^{\dagger}, 
\end{equation}
where $E_{0}$ and $E_{1}$ are the Kraus operators that describe the action of the channel. It is easy to see that the operators $E_{0}, E_{1}$ have the following form when written in the $\{|0\rangle, |1\rangle\}$ basis.
\begin{equation}
E_{0}  = \left( \begin{array}{cc}
1 & 0 \\
0 & f_{r,s}^{N}(t)
\end{array} \right), \; E_{1} = \left( \begin{array}{cc}
0 & \sqrt{1-|f_{r,s}^{N}(t)|^{2}} \\
0 & 0
\end{array} \right). \label{eq:Kraus_ideal}
\end{equation}
The Kraus operators in Eq.\eqref{eq:Kraus_ideal} lead to a channel that has the same structure as the amplitude damping channel, but is more general since the parameter $f_{r,s}^{N}(t)$ characterizing the noise in the channel is complex. 

Recall that the standard amplitude damping channel is parameterized by a {\it real} noise parameter $p$ and is described by a  pair of Kraus operators, written in the $\{|0\rangle, |1\rangle\}$ basis as~\cite{Nielsen},
\begin{equation}\label{eq:amp_damp}
E^{\rm AD}_{0}  = \left( \begin{array}{cc}
1 & 0 \\
0 & \sqrt{1-p}
\end{array} \right), \; E^{\rm AD}_{1} = \left( \begin{array}{cc}
0 & \sqrt{p}\\
0 & 0
\end{array} \right).
\end{equation}
This is the quantum channel induced in the original state transfer protocol in~\cite{bose} where the Hamiltonian considered is a Heisenberg chain in the presence of an external field of the form $\vec{B}$ =$B \hat{z}$, that is,
\begin{equation}
\tilde{\cH} =  - \frac{J}{2}\sum_{\langle i,j \rangle}\vec{\sigma}^{i}\cdot\vec{\sigma}^{j} - B\sum_{i}\sigma_{z}.  \label{eq:Heisenberg_B}
\end{equation}
By choosing the intensity of the $\vec{B}$-field appropriately, it is possible to adjust the phase of the complex amplitude $f_{r,s}^{N}(t)$ to be a multiple of $2\pi$ and hence replace $f_{r,s}^{N}(t)$ by $\vert f_{r,s}^{N}(t)\vert$, thus obtaining the amplitude damping channel described in Eq.~\eqref{eq:amp_damp} above. 

While much of the past work on state transfer has focused on the Heisenberg Hamiltonian in Eq.~\eqref{eq:Heisenberg_B}, here, we will focus on the more general Hamiltonian in Eq.~\eqref{eq:H_gen}. We study the problem of transmitting an arbitrary quantum state from the $s^{\rm th}$ site to the $r^{\rm th}$ site of an $N$-spin chain. We quantify the performance of the protocol in terms of the fidelity between the final state $\rho_{\rm out} \equiv \cE(|\psi_{\rm in}\rangle\langle\psi_{\rm in}|)$ and the input state $|\psi_{\rm in}\rangle$. Specifically, we use the {\it worst-case} fidelity, which is defined as~\cite{Nielsen},
\begin{equation}
 F^{2}_{\rm min} (\cE )  = \min_{a,b} \, \langle \psi_{\rm in} \vert \rho_{\rm out} \vert \psi_{\rm in} \rangle, \nonumber
\end{equation}
where the minimization is over all possible input states $a|0\rangle + b |1\rangle$. We say that pretty good state transfer is achieved when the worst-case fidelity $F^{2}_{\rm min}(\cE) \geq 1 -\epsilon$, for some $\epsilon > 0$. 

Let $|f^{N}_{r,s}(t)|$ and $\Theta$ refer to the amplitude and phase respectively, of the noise parameter $f^{N}_{r,s}(t) = e^{i\Theta}|f^{N}_{r,s}(t)|$ of the general quantum channel in Eq.~\eqref{eq:Kraus_ideal}. For such a channel, the worst-case fidelity depends on both the amplitude $|f^{N}_{r,s}(t)|$ as well as the phase $\Theta$. However, following the original protocol in~\cite{bose}, if we choose the magnetic fields $\{B_{k}\}$ so as to ensure that $\Theta$ is a multiple of $2\pi$, we can show that,
\begin{equation}
F^{2}_{\rm min} (\cE ) = |f_{r,s}^{N}(t)|^{2}. \label{eq:fmin_noQEC}
\end{equation}
In what follows, we examine how the worst-case fidelity may be improved using techniques from quantum error correction. In particular, by obtaining a functional relationship between the worst-case fidelity and the transition amplitude using an adaptive QEC procedure, we show how the fidelity can be improved by an order in the noise parameter.


\section{State Transfer protocol based on adaptive QEC}\label{sec:ideal}

Given a specific form of the spin-conserving Hamiltonian in Eq.~\eqref{eq:H_gen}, it is possible to estimate $|f^{N}_{r,s}(t)|$ and $\Theta$ for a specific choice of sites $s,r$ and $t$ by making repeated measurements on the spin chain~\cite{perfect}. Knowing $\Theta$, we may apply a phase gate of the form,
\begin{equation}
 U_{\Theta} = \left( \begin{array}{cc}
1 & 0 \\
0 & e^{-i\Theta} 
\end{array} \right),  \label{eq:theta-gate}
\end{equation}
to change the encoding basis to $\{|0\rangle, e^{-i\Theta}|1\rangle\}$. In this rotated basis, the channel in Eq.~\eqref{eq:Kraus_ideal} is identical to the amplitude damping channel described in Eq.~\eqref{eq:amp_damp}. At the level of the Hamiltonian, this is the same as choosing the field strengths $\{B_{k}\}$ so as to make the phase $\Theta$ trivial. Indeed, by making an appropriate choice of magnetic fields, it is always possible to transfom the spin-preserving Hamiltonian in Eq.~\eqref{eq:H_gen} into an $XXX$ interaction as in Eq.~\eqref{eq:Heisenberg_B}  (see~\cite{kayreview}) and hence map the underlying noise channel to an amplitude-damping channel. 

One na\"ive approach to improving the fidelity of state transfer is to therefore first apply the $U_\Theta$-gate and then use any of the well known QEC protocols which correct for amplitude damping noise~\cite{Leung, HKN_PM2010, Fletcher, AD_reliable2017}.  However, such an approach fails in the presence of disorder. When we consider a disordered $1$-d spin chain wherein either the  couplings $\{J_{k}, \tilde{J}_{k}\}$ or the fields $\{B_{k}\}$  in Eq.~\eqref{eq:H_gen} maybe random, the underlying noise channel is stochastic. The two real parameters $|f^{N}_{r,s}(t)|$ and $\Theta$ charcterizing the noise in the channel vary with each disorder realization, and hence an encoding procedure that relies on knowledge of a specific realization of $\Theta$ is not useful. Moreover, implementing a phase gate as in Eq.~\eqref{eq:theta-gate}based on the disorder-averaged value of $\Theta$ does not help -- such a phase gate will no longer cancel out the arbitrary (random) phase in Eq.~\eqref{eq:Kraus_ideal} and we do not obtain an amplitude damping channel in the rotated basis after the action of the phase gate.  

We would therefore like to tackle the problem of correcting for the more general noise channel in Eq.~\eqref{eq:Kraus_ideal}. Taking inspiration from the structural similarity to the amplitude damping channel, we propose a QEC protocol using an {\it approximate} $4$-qubit code~\cite{Leung} along with the channel-adapted near-optimal recovery proposed in~\cite{HKN_PM2010}. Specifically, we use a $4$-qubit code $\cC$, realized as the span of the following pair of orthogonal states,
\begin{eqnarray}\label{eq:4qubit}
|0_{L}\rangle &=& \frac{1}{\sqrt{2}}\left( \, |0000\rangle + |1111\rangle \, \right),\nonumber \\ 
|1_{L}\rangle &=& \frac{1}{\sqrt{2}}\left( \, |1100\rangle + |0011\rangle \, \right) .
\end{eqnarray}
This code was shown to be {\it approximately} correctable for amplitude damping noise, both in terms of worst case fidelity~\cite{Leung} as well as entanglement fidelity~\cite{cafaro2014simple}. The code is approximate in the sense it does not satisfy the conditions for perfect quantum error correction~\cite{Nielsen}, for any single-qubit error. 

The recovery map we use is adapted to the given noise map $\cE$ and code $\cC$, and can be described in terms of the Kraus operators of the noise and the projector $P$ onto the codespace, as follows,
\begin{equation}\label{eq:Petz}
\cR(.) =\sum_{i}P E_{i}^{\dagger}\cE(P)^{-1/2}(.)\cE(P)^{-1/2}E_{i}P ,
\end{equation}
where the inverse of $\cE({P})$ is taken on its support. Such a recovery map $\cR$ has been shown to achieve worst-case fidelity close to that of the optimal recovery map for any given noise channel $\cE$~\cite{HKN_PM2010}. In the specific case of the amplitude-damping channel and the $4$-qubit code, the adaptive recovery map defined above was shown to achieve better worst-case fidelity than the recovery used in~\cite{Leung}.

\begin{figure}[H]
\centering
\includegraphics[width=0.5\textwidth, height=.2\textwidth]{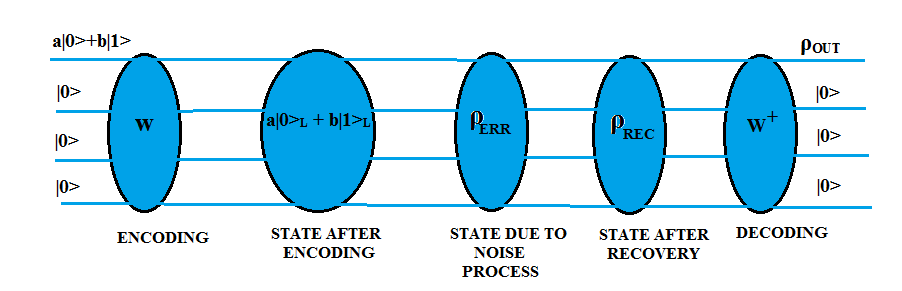}
\caption{$4$-qubit QEC on spin chains}
\label{fig:QEC_schematic}
\end{figure}

The quantum state transfer protocol with QEC is implemented using a set of $4$ unmodulated, identical, spin chains. Fig.~\ref{fig:QEC_schematic} depicts a schematic of our protocol. The initial, encoded state $|\psi_{\rm enc}\rangle$ is now an entangled state across the four chains, involving only a single spin (the $s^{\rm th}$ site)in each of the chains.
\begin{equation}\label{eq:psi_enc}
\vert\psi_{\rm enc}\rangle = a|0\rangle_{L} +b |1 \rangle_{L} .
\end{equation}
Once the initial state is prepared, the four chains are allowed to evolve in an uncoupled fashion, according to the Hamiltonian in Eq.~\eqref{eq:H_gen}. After time $t$, the state at the receiver's site is a joint state of the $r^{\rm th}$ site of the four chains, and is described by action of the map $\cE^{\otimes 4}$ with the time-dependent noise parameter $f_{r,s}^{N}(t)$. Thus,
\[ \rho_{\rm err} = \small{\cE^{\otimes 4}(\rho_{\rm enc})} = \sum_{i} E^{(4)}_{i} \rho_{\rm enc} \left( E^{(4)}_{i}\right)^{\dagger}, \]
where $E^{(4)}_{i}$ are the Kraus operators of the $4$-qubit noise channel realized as four-fold tensor products of the operators $E_{0}$ and $E_{1}$ in Eq.~\eqref{eq:Kraus_ideal}. After evolving the chains for time $t$, the recovery map $\cR^{(4)}$ is applied at the receiver's site of the four spin chains. The final state at the receiver's end after the QEC protocol is obtained as,
\begin{equation} 
\rho_{\rm rec} = \sum_{i,j} R^{(4)}_{j} E^{(4)}_{i} \rho_{\rm{enc}} \left(E^{(4)}_{i}\right)^{\dagger} \left(R^{(4)}_{j}\right)^{\dagger}, \nonumber
\end{equation}
with the Kraus operators $R^{(4)}_{i}$ given by, 
\begin{equation}
R^{(4)}_{i} = P\small{\left(E^{(4)}_{i}\right)^{\dagger}}\cE^{\otimes 4}(P)^{-1/2}, \label{eq:4qubit_petz}
\end{equation}
where $P\equiv |0_{L}\rangle\langle 0_{L}| + |1_{L}\rangle\langle 1_{L}|$ is the projector onto the $4$-qubit space described in Eq.~\eqref{eq:4qubit}. The fidelity of the $4$-chain quantum state transfer protocol is then given by,
\[F^{2}_{\rm min} \left( \cR^{(4)}\circ\cE^{\otimes 4}, \cC \right) \equiv \min_{a,b}\langle \psi_{\rm enc}\vert \rho_{\rm rec}\vert\psi_{\rm enc}\rangle.,\]
where the minimization is over all states in the codespace $\cC$. As before, pretty good transfer is achieved when the worst-case fidelity is high, that is, $F^{2}_{\rm min}\left( \cR^{(4)}\circ\cE^{\otimes 4}, \cC \right) \geq 1 - \epsilon$, for $\epsilon > 0$. We now present a key result of the paper, namely a bound on the fidelity of state transfer using the adaptive QEC protocol, in terms of the transition amplitude $f^{N}_{r,s,}(t)$.

\begin{theorem}\label{thm:aqec_fid}
The fidelity of quantum state transfer from site $s$ to site $r$ under a spin-conserving Hamiltonian as in Eq.~\eqref{eq:H_gen}, using the $4$-qubit code $\cC$ and adaptive recovery $\cR^{(4)}$ at time $t$, is given by,
\begin{equation}
F_{\rm min}^{2} \left( \cR^{(4)}\circ\cE^{\otimes 4}, \cC \right) \approx 1- \frac{7p^{2}}{4}+ O(p^{3}), \label{eq:4qubit_fid}
\end{equation}
where $p= 1-|f_{r,s}^{N}(t)|^{2}$. 
\end{theorem}
\begin{proof}
We first rewrite the Kraus operators given in Eq.~\eqref{eq:Kraus_ideal}, as, 
\begin{eqnarray}
 E_{0} &=& |0 \rangle \langle 0| +|f_{r,s}^{N}(t)|e^{i \Theta} |1\rangle \langle 1| \nonumber \\ \nonumber
 E_{1} &=& |0\rangle \langle 1| \sqrt{1-|f_{r,s}^{N}(t)|^{2}} ,
\end{eqnarray}
where, $|f_{r,s}^{N}(t)|$ and $\Theta$ are the absolute value and phase of the complex-valued transition amplitude $f_{r,s}^{N}(t)$. The state after the $4$-qubit recovery map is then given by,
\[ \rho _{\rm rec} = \left(\small{\cR^{(4)}\circ\cE^{\otimes 4}}\right) \left(\rho_{\rm enc}\right).\]
The composite map $\left(\cR^{(4)}\circ\cE^{\otimes 4}\right)$ comprising noise and recovery has Kraus operators of the form,
\begin{equation}
 \small{P\left(E^{(4)}_{j}\right)^{\dagger}}\cE^{\otimes 4}(P)^{-1/2}E^{(4)}_{i}P. \label{eq:kraus_composite}
 \end{equation}
The key step in obtaining the desired fidelity is to show that the Kraus operators of the composite map written above are independent of $\Theta$. First, we write out $\small{\cE^{\otimes 4}(P)^{-1/2}}$ in the (standard) computational basis of the $4$-qubit space. 
\begin{eqnarray}
&& \small{\cE^{\otimes 4}(P)^{-1/2}} = \sum_{i=1}^{16}\cG_{i}|i\rangle\langle i| + e^{-4 i \Theta } \cG_{17}|0000\rangle\langle 1111| \nonumber \\
 &+& e^{i 4\Theta} \cG_{17}|1111\rangle\langle0000| + \cG_{18}(|1100\rangle \langle0011| +|0011\rangle \langle 1100|), \nonumber
\end{eqnarray}
where $\{\cG_{i}\}$ are polynomial functions of the transition amplitude $|f_{r,s}^{N}(t)|$. The $\Theta$-dependence in this pseudo-inverse operator occurs only in the span of $\{|0000\rangle, |1111\rangle\}$. Since $\small{\cE^{\otimes 4}(P)^{-1/2}}$ is sandwiched between the Kraus operators of the $4$-qubit channel and their adjoints, we also write down the Kraus operators $\{E^{(4)}_{i}\}$ in the computational basis. Then, an explicit computation reveals that the $\Theta$-dependence gets conjugated out for each of the Kraus operators in Eq.~\eqref{eq:kraus_composite}.  We refer to Appendix~\ref{sec:E(P)} for the details of this calculation.


\noindent Hence the final state after noise and recovery $\rho_{\rm rec}$ can be expressed as a linear sum of terms that are independent of $\Theta$. Since the parameter $\Theta$ is effectively suppressed, the fidelity after using $4$-qubit code and the universal recovery in Eq.~\eqref{eq:Kraus_ideal}, is purely a function of $p = 1-|f_{r,s}^{N}(t)|^{2}$.  

\noindent The fidelity corresponding to the initial state $|\psi_{\rm enc}\rangle = a|0_{L}\rangle + b|1_{L}\rangle$ can thus be obtained as,
\begin{eqnarray} \label{eq:4qubitmin}
&& F^{2} ( \cR^{( 4)}\circ\cE^{\otimes 4},\cC)  \nonumber \\
&=& 1 - p^{2}\left((|a|^{2} - |b|^{2})^{2} - ((ba^{*})^{2}+(ab^{*})^2)+ 5 |a|^{2}|b|^{2}\right) \nonumber \\
&& + \, O(p^{3}) ,
\end{eqnarray}
where $O(p^{3})$ refers to terms of order $p^{3}$ and higher. Parameterizing $a$ and $b$ as $a=\cos{\frac{\theta}{2}}$, $b=e^{-i \phi}\sin{\frac{\theta}{2}}$, the fidelity attains its minimum value at $\{\theta,\phi\}= \{\frac{(2n+1)\pi}{2},\frac{(2n+1)\pi}{2}\}$ ($n=1,2,\ldots$), so that the worst-case fidelity over the $4$-qubit code $\cC$ is given by,
\[
F_{\rm min}^{2} \left( \cR^{(4)}\circ\cE^{\otimes 4}, \cC \right) \approx 1- \frac{7p^{2}}{4}+ O(p^{3}).
\] 
\end{proof}

Our result shows that using the adaptive recovery in conjunction with the approximate code leads to a fidelity that is independent of the phase $\Theta$ of the complex noise parameter $f^{N}_{r,s}(t)$. Thus, to optimize the fidelity of state transfer between the $s^{\rm th}$ and $r^{\rm th}$ site of a chain of $N$ spins evolving according to the Hamiltonian in Eq.~\eqref{eq:H_gen}, we simply need to find the time $t$ at which $\vert f^{N}_{r,s}(t)\vert^{2}$ is maximized. Recall that the worst-case fidelity without QEC (using the single chain protocol) is linear in the parameter $p$, as observed in  Eq.~\eqref{eq:fmin_noQEC}. Thus we see an $O(p)$ improvement in fidelity with QEC, as expected. 

Furthermore, our estimate of the worst-case fidelity implies that so long as the noise strength $p$ is such that $1 - (7/4)p^{2} > 1- p$, the adaptive QEC protocol achieves better fidelity than the single chain protocol without QEC. This constraints the noise strength $p$ to satisfy $0<p< (4/7) \approx 0.57$. This in turn implies a threshold for the transition amplitude, namely, $|f_{r,s}^{N}(t)|^{2} > 0.43$, below which our adaptive QEC protocol will not offer any improvement in the fidelity of state transfer. 


\section{Results for the $1$-d Heisenberg Chain} \label{sec:Heisenberg}

As a simple example to illustrate the performance of the adaptive QEC protocol, we now consider a special case of the Hamiltonian in Eq.~\eqref{eq:H_gen}, namely, an $N$-length, ideal Heisenberg chain, with $J_{k} = \tilde{J}_{k} =  J/2 (J>0)$ and $B_{k} =0$, for all $k$. This is also often referred to as the $XXX$-chain in the literature. Setting $J=1$ without loss of generality, we present numerical results on the fidelity of state transfer from the first ($s=1$) to the $N^{\rm th}$ ($r=N$) site.
\begin{figure}[H]
\centering
\includegraphics[width=0.5\textwidth]{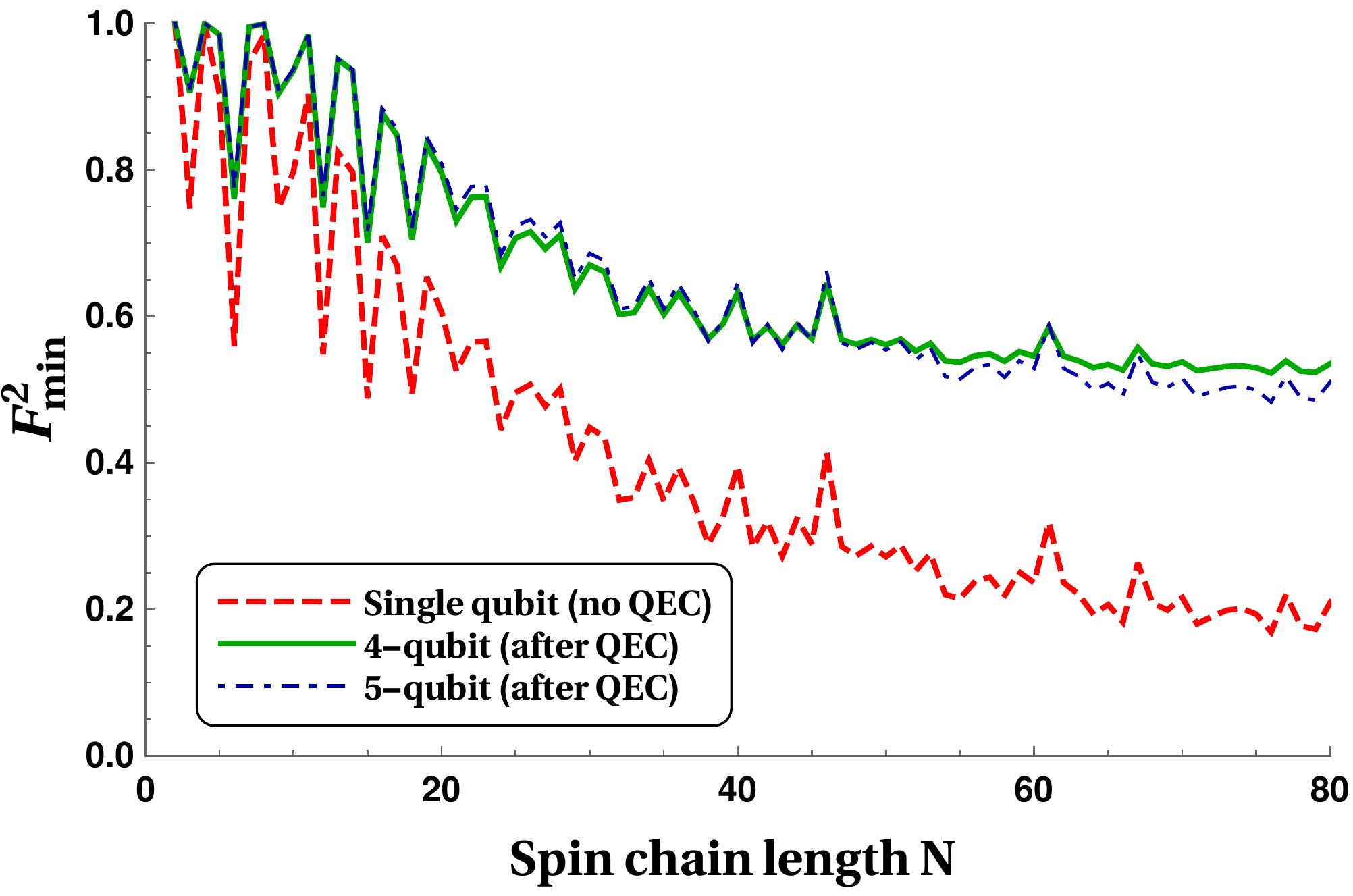}
\caption{Worst-case fidelity as a function of chain length $N$.}
\label{fig:fmin_QEC}
\end{figure}
Fig.~\ref{fig:fmin_QEC} compares the performance of state transfer protocols with and without QEC. In particular, it compares the performance of our $4$-chain state transfer protocol with the single-chain (no QEC) protocol~\cite{bose} and the $5$-chain protocol proposed in~\cite{allcock}. For each $N$, we plot the fidelity of state transfer from the $1^{\rm st}$ site to the $N^{\rm th}$ site on a $N$-length spin chain, after a time $t^{*}$ chosen such that $|f_{N,1}^{N}(t)|$ is maximum at $t=t^{*}$, for $0<t<4000/J$.

From the plot we see that the QEC-based protocols achieve pretty good state transfer over longer distances than the single chain protocol. Furthermore, using {\it approximate} QEC it is possible to achieve as high as fidelity as with the standard $5$-qubit code, using fewer spin chains. Specifically, in the regime of small noise parameter $p$, we can show that the worst-case fidelity obtained using the $5$-qubit code is,
\begin{equation}
F^{2}_{\min} \approx 1-\frac{15p^{2}}{8}+ O(p^{3}) . \label{eq:5qubit_fid}
\end{equation}
Correspondingly, a $5$-chain protocol performs well over single chain protocol when $0< p < (8/15) \approx 0.53$, implying that the transition amplitude should satisfy $|f_{r,s}^{N}(t)|^{2} > 0.47$, which is a higher threshold than that required by our adaptive QEC protocol.

For the ideal Heisenberg chain, it was recently shown that~\cite{godsil}, there always exists a time $t$ at which $|f_{1,N}^{N}(t)|^{2} > 1 - \epsilon$ if and only if the length of the chain is a power of $2$, that is, $N = 2^{m}$. In other words, pretty good state transfer is always possible between the ends of a Heisenberg spin chain whose length $N$ is of the form $N = 2^{m} (m>1)$.  We may therefore consider improving the performance of our QEC-based protocol by repeating the error correction procedure every $2^{m}$ sites. Specifically, we can achieve pretty good state transfer over a chain of arbitrary length $L$, by stitching together smaller chains whose lengths are of the form $N =2^{m}$. At every stage of the repeated QEC protocol, there are exactly $2^{m}$ interacting spins and the rest of the spin-spin-interactions are turned off. 
 \begin{figure}[H] 
 \centering
 \includegraphics[width=0.47\textwidth]{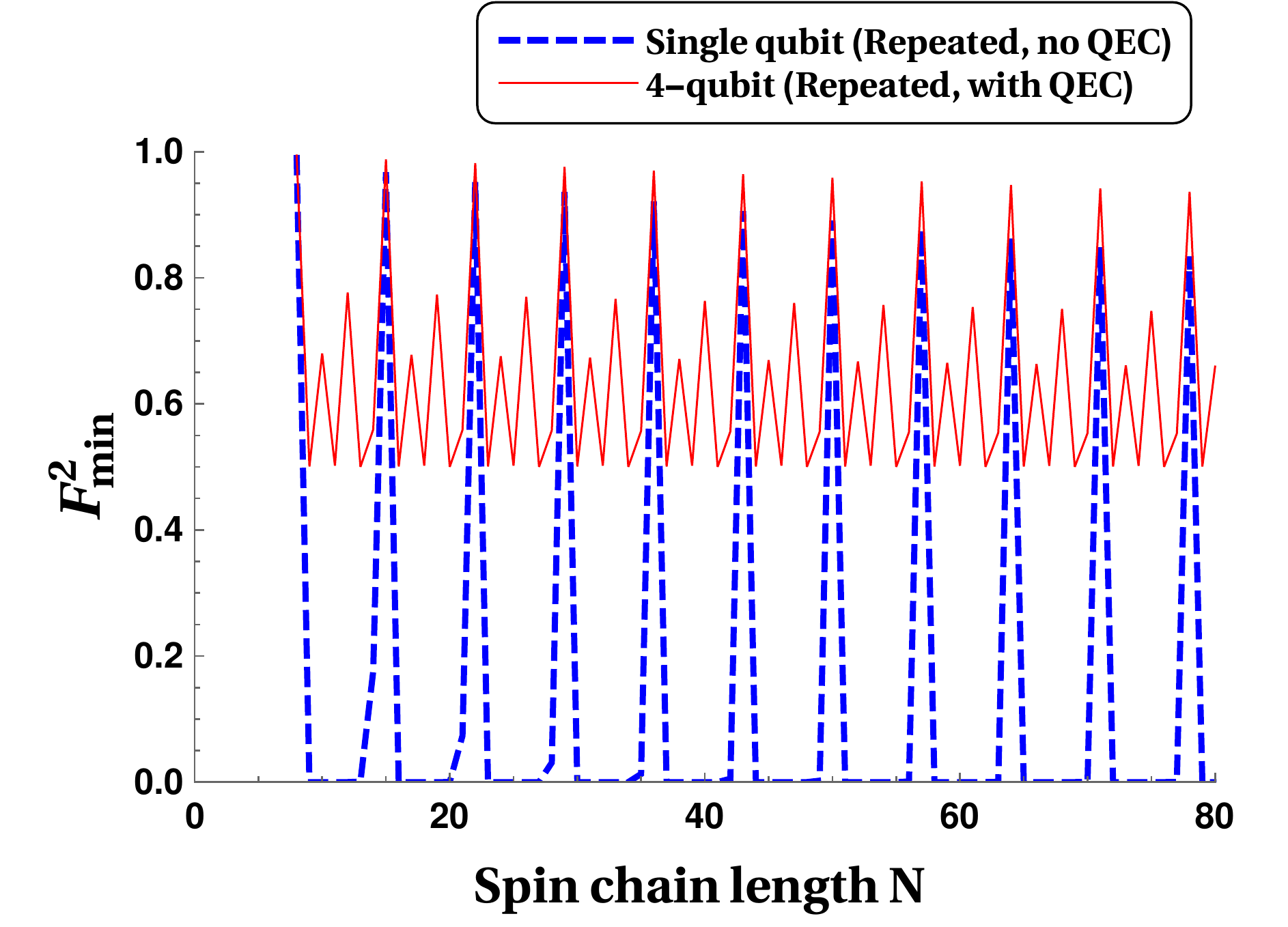}
 \caption{Worst-case fidelity using repeated QEC}
 \label{fig:repeated_QEC}
 \end{figure}
Fig.~\ref{fig:repeated_QEC} shows an example of the resulting improvement in fidelity when the QEC protocol is repeated every $8$ sites. For comparison, we plot the worst-case fidelity obtained by stitching together a sequence of length-$8$ chains, without QEC. The repeated QEC protocol proceeds as follows. We first implement our QEC protocol for an $8$-spin chain, evolving for time $t^{*}$ at which $|f_{8,1}^8 (t)|$ maximizes. We repeat this procedure some $k$ times, where $k$ is the largest integer such that $8k < N$ and finally perform QEC for the remaining $N-7k$ sites for the same waiting time $t^{*}$. Such a repeated QEC protocol indeed enables pretty good transfer for much  longer lengths, as seen in the plot.  

More generally, if $F^{2}_{\rm min} \approx 1 - \alpha p^{2}$ is the fidelity of the single-shot QEC protocol, repeating the procedure $k$ times gives us a fidelity of $F^{2}_{\min} = 1-(p_{\rm{new}})$ , with,
\begin{eqnarray}
p_{\rm{new}}= (1-(1-\alpha p^{2})^{k}) \nonumber
\end{eqnarray}
where $p_{\rm{new}}$ is the noise parameter obtained after repeating QEC $k$ times.

\section{Quantum state transfer on a disordered Heisenberg chain}\label{sec:disorder}

Moving away from an ideal spin chain with a fixed, uniform coupling between successive spins, we now study state transfer over a disordered $XXX$ chain, where the spin-spin couplings are randomly drawn from some distribution. It is well known that the presence of disorder in a $1$-d spin chain leads to the phenomenon of localization~\cite{anderson} of information close to one end of the chain. It is therefore a challenging task to identify protocols which achieve perfect or pretty good transfer over disordered spin chains, overcoming the effects of localization. 

Past work on disordered chains has primarily focused on the $XX$ chain. Starting with a modulated chain that admits perfect state transfer, both random magnetic field and random couplings have been studied~\cite{Chiara}. Alternately, an unmodulated chain with random couplings at all except the sender and receiver sites has also been studied~\cite{ashhab}. 

When viewed in the quantum channel picture, the presence of disorder becomes as an additional source of noise. The role of QEC in overcoming the effects of disorder has been studied both for the $XX$~\cite{kay} as well as the Heisenberg chains~\cite{allcock}. The QEC protocol for a noisy $XX$ chain with random couplings involves encoding into multiple spins on a single chain using modified CSS codes~\cite{kay}. The QEC protocol in~\cite{allcock} encodes into multiple identical, uncoupled chains using the standard $5$-qubit code, while also requiring access to multiple spins at the sender and receiver ends of each of the chains. Furthermore, the protocol based on the $5$-qubit code involves choosing an encoding based on the phase $\Theta$ of the transition amplitude (as explained in Sec.~\ref{sec:ideal}), which in turn is specific to the disorder realization. This makes the QEC procedure hard to implement in a practical sense. 

Here, we show how the channel-adapted QEC procedure described in Sec.~\ref{sec:ideal} can be used to achieve pretty good state transfer over an $XXX$ chain with random couplings. As before, we quantify the performance of the state transfer protocol in terms of the fidelity between the initial and final states. When the underlying quantum channel is stochastic, as in the case of a disordered chain, we use the {\it disorder-averaged} worst-case fidelity $\langle F^{2}_{\rm min}\rangle_{\delta}$, to characterize the performance of the state transfer protocol. We say that pretty good state transfer is achieved by a certain choice of code $\cC$ and recovery $\cR$ when the corresponding disorder-averaged fidelity $\langle F^{2}_{\rm min}\rangle_{\delta} \geq 1 -\epsilon$, for some $\epsilon > 0$. 

We consider a disordered Heisenberg chain with couplings $J_{k}$= $\frac{\overline{J}}{2}(1+\Delta_{k})$, where $\Delta_{k}$ are independent, identically distributed random variables drawn from a uniform distribution between $ \left[ -\delta,\delta \right ] $ and $\overline{J}$ is the mean value of the coupling strength, which we may set to $1$, without loss of generality.  Note that such a Hamiltonian conserves the total spin and hence falls within the universality class discussed in Sec.~\ref{sec:prelim}. 

Consider a state transfer protocol, where the sender wishes to transmit the state $|\psi_{\rm in}\rangle = a|0\rangle + b|1\rangle$ from the $s^{\rm th}$ site to the $r^{th}$ site via the natural dynamics of the chain. As before, the final state at the receiver's site, tracing out the other spins can be realized as the action of a quantum channel $\cE$,
\[ \rho_{\rm out} = \cE(\rho_{\rm in}) = \sum_{k}E_{k}\rho_{\rm in} E_{k}^{\dagger}, \]
with the same Kraus operators $\{E_{0}, E_{1}\}$ as in Eq.~\eqref{eq:Kraus_ideal}. The key difference however is in the nature of the noise parameter $p \equiv 1 - |f_{r,s}^{N}(t,\{\Delta_{k} \})|^{2}$: in the case of the {\it disordered} chain, the transition amplitude $f_{r,s}^{N}(t, \{\Delta_{k}\} ) $ between site $s$ and $r$ for a chain of length $N$ allowed to evolve for a time $t$, is a random variable whose value depends on the specific realization of the disorder variables $\{\Delta_{k}\}$. The distribution of $f_{r,s}^{N}(t, \{\Delta_{k}\} )$ for given set of $r,s,N,t$ values depends on the distribution over which the disorder variables $\{\Delta_{k}\}$ are sampled. To illustrate our point, we specifically consider below the case where the coupling strengths $\{\Delta_{k}\}$ are independently sampled from a uniform distribution.

\subsection{Transition amplitude in the presence of disorder}\label{sec:transAmp_disorder}

The Heisenberg Hamiltonian $\mathcal{H}$ with static disorder in the coupling strengths, has the form,
\begin{equation}\label{eq:H_dis}
\mathcal{H}_{\rm dis} = -\sum_{k}\frac{\overline{J}(1+\Delta_{k})}{2}(\sigma^{k}_{x}\sigma^{k+1}_{x}+\sigma^{k}_{y}\sigma^{k+1}_{y}+ \sigma_{z}^{k}\sigma^{k+1}_{z}).
\end{equation}
Here, the effect of disorder is introduced via the i.i.d. random variables $\{\Delta_{i}\}$ which take values over a uniform distribution between $\left [-\delta,\delta \right ]$. The quantity $\delta$ is called the disorder strength, and  $\overline{J}$ is the mean value of the coupling coefficient. We may view the disordered Hamiltonian as a sum of the form $\mathcal{H}_{\rm dis} =\mathcal{H}_{o}+ \mathcal{H}_{\delta}$, where $\mathcal{H}_{o}$ denotes the ideal $XXX$ Hamiltonian studied in the previous section and $\mathcal{H}_{\delta}$ is given by,
\[ \mathcal{H}_{\delta} = -\frac{\overline{J}}{2} \sum_{k}\Delta_{k} \overrightarrow{\sigma^{k}} \cdot \overrightarrow{\sigma^{k+1}}. \]
$\mathcal{H}_{\delta}$ captures the effect of disorder in the spin chain and can be treated as a perturbation of the Hamiltonian $\cH_{0}$. Since $[\cH_{0},\cH_{\delta}] \neq 0$, the transition amplitude maybe evaluated using the so-called time-ordered expansion, also referred to as the Dyson-series~\cite{dyson}. 

Specifically, the transition amplitude between the $r^{\rm th}$ and $s^{\rm th}$ site for the disordered Hamiltonian $\mathcal{H}_{\rm dis}$ in Eq.~\eqref{eq:H_dis} is given by (setting $\hbar = 1$),
\begin{eqnarray}
&& f^{N}_{r,s}(t, \{\Delta_{k}\} \,) \nonumber \\
&=& \langle \textbf{r} | e^{- i (\mathcal{H}_{o}+ \mathcal{H}_{\delta})t} |\textbf{s}\rangle \nonumber \\ 
&=& \langle \textbf{r}| e^{-i \mathcal{H}_{o} t}\mathcal{ T}\left[\exp{\left(-i\int_{0}^{t}e^{i\mathcal{ H}_{o} t'}\,\mathcal{H_{\delta}}\,e^{-i \mathcal{H}_{o} t'}dt'\right)}\right]| \textbf{s}\rangle \nonumber \\  
&=& f^{N}_{r,s}(t) - i\sum_{k=1}^{N} f^{N}_{r,k}(t) \int_{0}^{t}\langle \textbf{k}|e^{i \mathcal{H}_{o} t'} \mathcal{H}_{\delta}e^{-i\mathcal{ H}_{o} t'}|\textbf{s}\rangle dt' \nonumber \\
&& + \;  O(H_{\delta}^{2}), \nonumber
\end{eqnarray}
where $\cT$ is the time-ordering operator which has been expanded to first order in the perturbation in the final equation. As before, $f^{N}_{r,k}(t)$ denotes the transition amplitude between the $r^{\rm th}$ and $k^{\rm th}$ sites in the case of an ideal chain of length $N$, without disorder.  

Thus, using the time-ordered expansion, the transition amplitude in the presence of disorder can be evaluated as a perturbation around the zero-disorder value $f^{N}_{r,s}(t)$, of the form,
\begin{equation}
 f^{N}_{r,s}(t,\{\Delta_{k}\}) = f^{N}_{r,s}(t)+ \sum_{i=1}^{N-1} c^{N}_{i}(t) \Delta_{i} + \sum_{i,j=1}^{N-1}d^{N}_{ij} \Delta_{i}\Delta_{j} + \ldots . \label{eq:transAmp_final}
 \end{equation}
The explicit forms of the complex coefficients $c^{N}_{i}(t)$ are given in Eq.~\eqref{eq:c-coeff} in Appendix~\ref{sec:transAmp_dist}. A similar approach was used in~\cite{Chiara} to study deviations from perfect state transfer due to the presence of disorder in an $XX$ chain.

Using the form of the transition amplitude stated in Eq.~\eqref{eq:transAmp_final}, we obtain the distribution of real part of the transition amplitude $x \equiv \texttt{Re}[f^{N}_{r,s}(t,\{\Delta_{k}\})]\,$, up to first order in the perturbation $H_{\delta}$, as,
\begin{equation}
 \cP^{\delta, N,t}(x) \propto \sum_{j=1}^{2^{N-1}} (-1)^{u_{j}}(q_{j})^{N-2}\,{\rm Sign}[q_{j}], 
 \end{equation}
where  $u_{j} \in [0,1]$ and the $Sign$ function is defined as
\begin{equation}
 {\rm Sign }(x-a) = \left\lbrace \begin{array}{cc}
   -1 , & x< a , \\ 
   0 , & x=a , \\ 
   1, & x>a . \end{array} \right. \nonumber
\end{equation}
The functions $q_{j} \left(x,\texttt{Re}[f^{N}_{r,s}(t)], \{\texttt{Re}[c^{N}_{i}(t)]\} \right)$ are linear combinations of the form,
\begin{equation}
 q_{j} \equiv x - \texttt{Re}[f^{N}_{r,s}(t)] + \delta\sum_{i=1}^{N-1} (-1)^{r_{i}^{j}}\texttt{Re}[c^{N}_{i}(t)]  , 
 \end{equation}
where $ r_{i}^{j} \in [0,1] \, \forall \, i=1,\ldots, N-1$ and $\texttt{Re}[c^{N}_{i}(t)]$ denote the real part of the coefficients in Eq.~\ref{eq:transAmp_final}. Since there are $N-1$ such coefficients for a spin chain of length $N$, the sum over $i$ ranges from $1$ to $N-1$. There are $2^{N-1}$ distinct linear combinations of the form $q_{j}$, corresponding to the $2^{N-1}$ distinct $(N-1)$-bit binary strings parameterized by $r^{j}$, so that the sum over $j$ runs from $1$ to $2^{N-1}$. The form of the distribution is identical for the imaginary part $\texttt{Im}[f^{N}_{r,s}(t,\{\Delta_{k}\})]$, with the real parts of $\{c^{N}_{i}(t)\}$ and $f^{N}_{r,s}(t)$ replaced by their imaginary parts. We refer to Eqs.~\eqref{eq:dist_real2},~\eqref{eq:dist_im2} in Appendix~\ref{sec:transAmp_dist} for a detailed description of the distributions of the real and imaginary parts of the disordered transition amplitude.

\begin{figure} 
\centering
\includegraphics[width=0.47\textwidth]{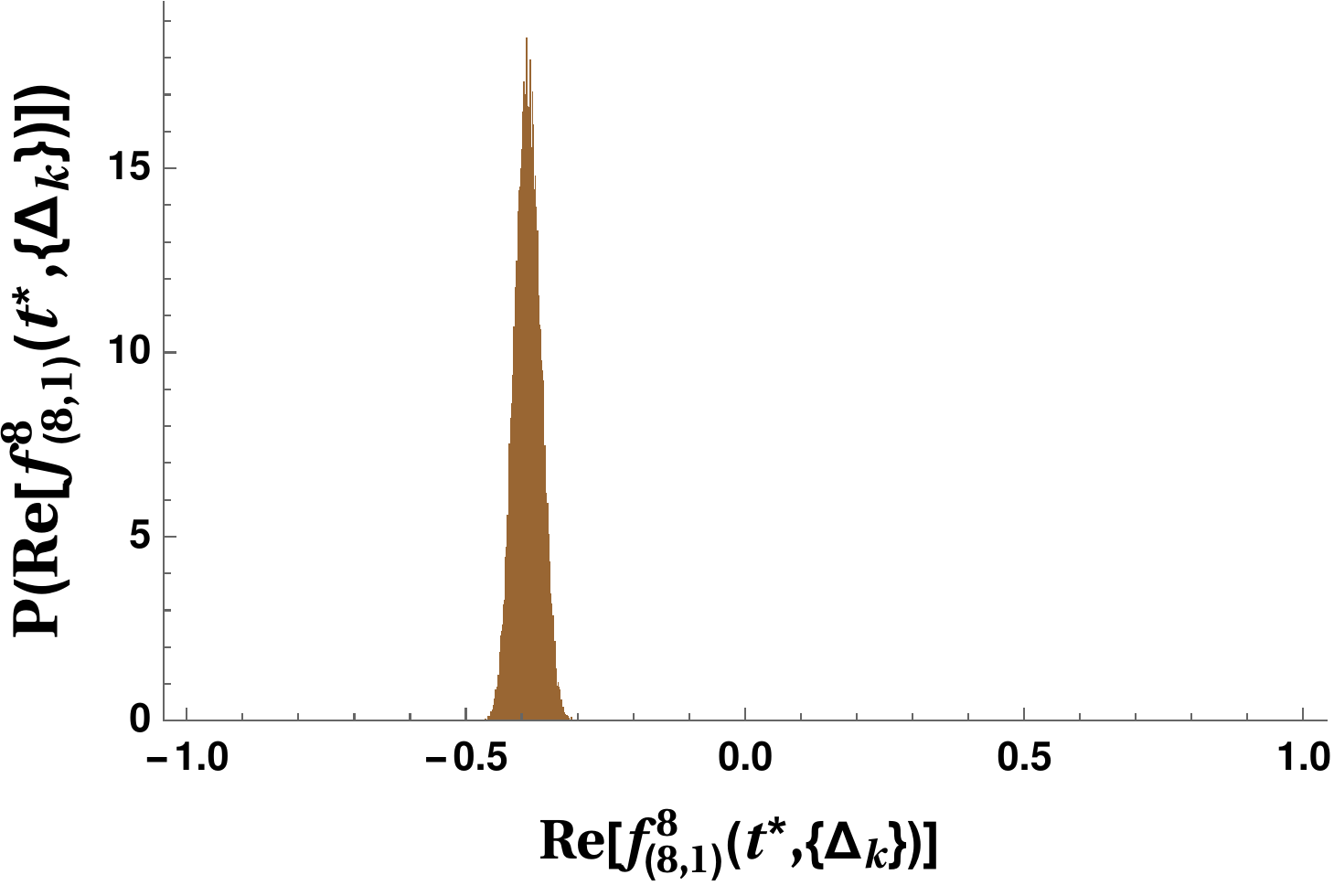}

\caption{Distribution of $\texttt{Re}[f_{8,1}^{8}(t^{*}, \{\Delta_{k}\})]$ 
for different disorder realizations, drawn from a uniform distribution with disorder strength $\delta = 0.001$.}
\label{fig:f_dist1}
\end{figure}

The key salient feature we observe from calculating the distribution functions above is that the limiting distribution in the case of no disorder ($\delta \rightarrow 0$), is indeed a delta distribution peaked around $f^{N}_{r,s}(t)$. Furthermore, in Appendix~\ref{sec:transAmp_dist} we also explicitly evaluate the mean and variance of $f^{N}_{r,s}(t,\{\Delta_{k}\})$ and show that the mean is equal to the zero-disorder value of $f^{N}_{r,s}(t)$, up to $O(\delta^{2})$ (see Eq.~\ref{eq:trans_ampAvg}). The variance goes as $O(\delta^{2})$, as shown in Eq.~\eqref{eq:variance}, making it vanishingly small in the limit of small $\delta$. This observation leads us to propose a modified QEC protocol for state transfer over disordered $XXX$ chains, using an adaptive recovery $\cR_{\rm avg}$ based on the {\it disorder-averaged} transition amplitude $\langle f_{r,s}^{N}(t,\{\Delta_{k}\})\rangle_{\delta}$.

\begin{figure}  
\centering
\includegraphics[width=0.47\textwidth]{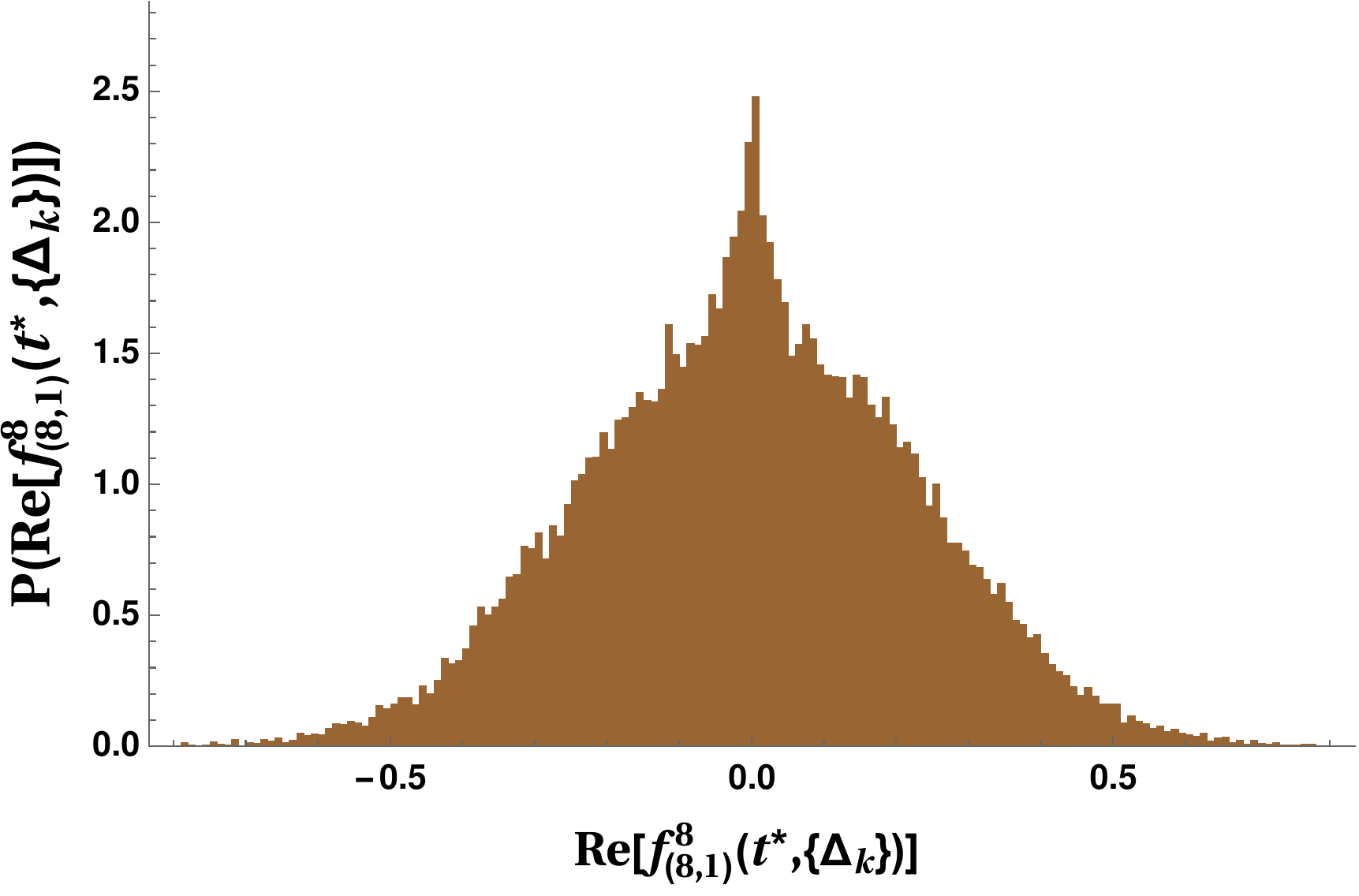}

\caption{Distribution of $\texttt{Re}[f_{8,1}^{8}(t^{*}, \{\Delta_{k}\})]$ 
for different disorder realizations, with disorder strength $\delta = 1$.}
\label{fig:f_dist2}
\end{figure}

The analysis presented thus far holds for any pair of sites $(s,r)$ on a spin chain of length $N$. As an example, we consider the specific case of an $8$-length chain, with $s=1$ and $r=8$. We plot the distribution of the real part of the transition amplitude at some fixed time $t^{*}$, for disorder strengths $\delta = 0.001$ and $\delta=1$, in Figs.~\ref{fig:f_dist1} and~\ref{fig:f_dist2} respectively. We see that when the disorder strength is small enough, the transition amplitude is indeed distributed like a delta function peaked around the zero-disorder value. For large values of $\delta$, the distribution spreads out quite a bit and its mean also shifts closer to zero, giving rise to a very small transition amplitude. The corresponding figures for $\texttt{Im}[f_{8,1}^{8}(t^{*}, \{\Delta_{k}\})]$ are presented in Fig.~\ref{fig:f_dist_Im}. 

\subsection{Adaptive QEC for $1$-d disordered chain}\label{sec:aqec_diorder}
To summarize, the quantum channel for state transfer in the presence of disorder has the same structure as that of the ideal chain, but with a stochastic noise parameter $p \equiv 1 - |f^{N}_{r,s}(t,\{ \Delta_{k} \})|^{2}$, since the transition amplitude $f_{r,s}^{N}(t,\{\Delta_{k}\})$ is now a random variable whose value depends on the random couplings $\{\Delta_{k}\}$. However, as discussed in Sec.~\ref{sec:transAmp_disorder}, for small enough disorder strengths, $f_{r,s}^{N}(t,\{\Delta_{k}\})$ is peaked sharply around its mean value, and we may consider the disorder-averaged amplitude $\langle f_{r,s}^{N}(t,\{\Delta_{k}\})\rangle_{\delta}$ as a good estimate of the noise. 

We therefore propose an adaptive QEC procedure for a disordered $XXX$ chain involving the $4$-qubit code in Eq.~\eqref{eq:4qubit} and a recovery map $\cR_{\rm avg}$ with the same structure as that used in the case of the ideal chain, described in Eq.~\eqref{eq:4qubit_petz}. However, unlike the ideal case, the value of the channel parameter used in the recovery is different from the one in actual noise channel :  the recovery map uses the disorder-averaged amplitude $\langle f_{r,s}^{N}(t,\{\Delta_{k}\})\rangle_{\delta}$, and is therefore independent of the specific disorder realization, whereas the noise channel has the parameter $ f_{r,s}^{N}(t,\{\Delta_{k}\})$ which changes with every realization.

\begin{figure} [h!]
\centering
\includegraphics[width=0.47\textwidth]{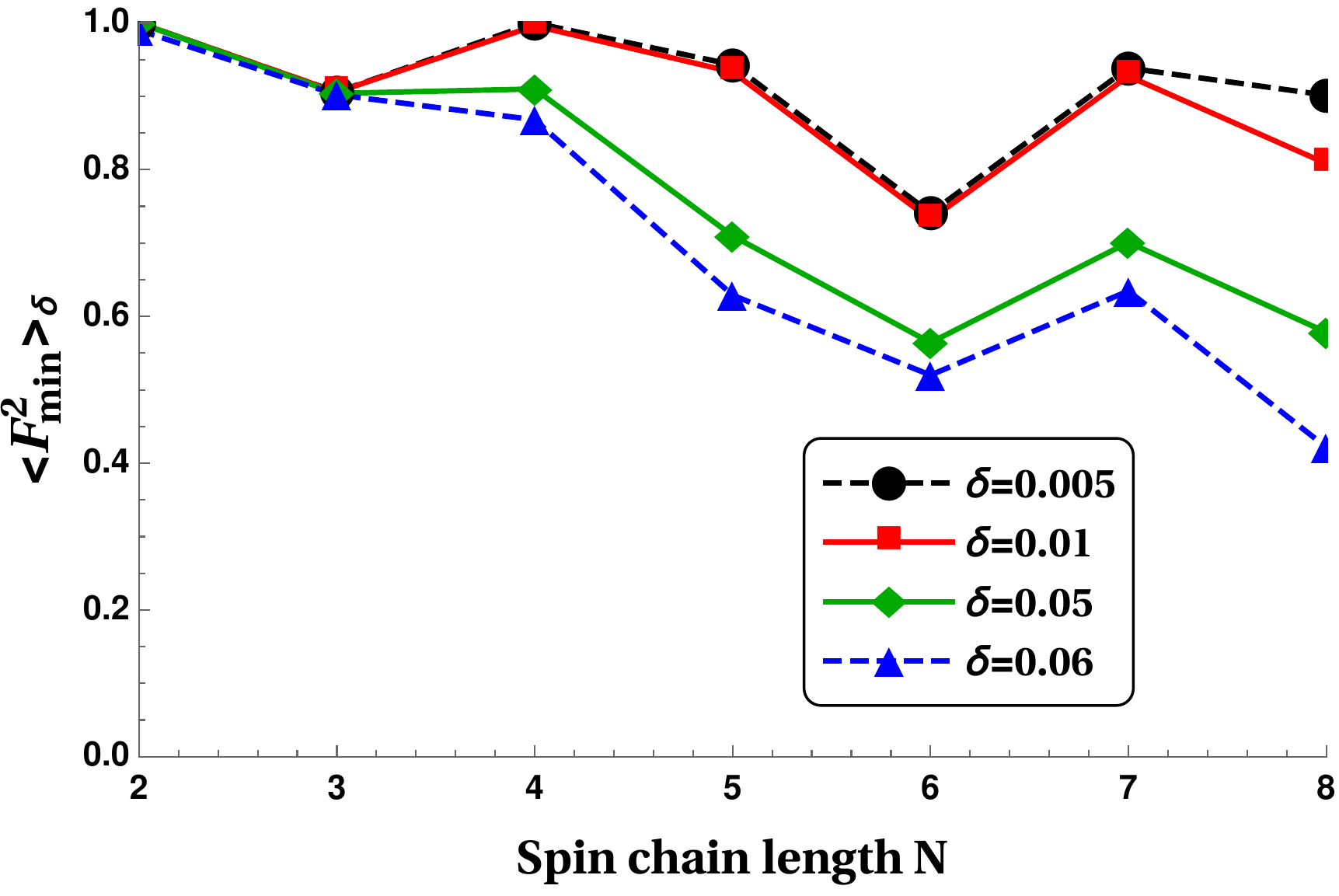}
\caption{Disorder-averaged worst-case fidelity $\langle F^{2}_{\rm min}\rangle_{\delta}$ obtained using the adaptive recovery $\cR_{\rm avg}$.}
\label{fig:disavg_fmin_8}
\end{figure}
To illustrate the performance of this modified recovery map, we present numerical results for quantum state transfer from the first site ($s=1$) to the $8^{\rm th}$ site ($r=8$)on an $8$-spin chain. Fig.~\ref{fig:disavg_fmin_8} shows the disorder-averaged worst-case fidelity $\langle F^{2}_{\rm min}\rangle_{\delta}$ obtained using the $4$-qubit code and the adaptive recovery $\cR_{\rm avg}$, for an $8$-spin chain. For disorder strengths $\delta \leq 0.01$,  our adaptive QEC protocol achieves pretty good transfer, with fidelity-loss $\epsilon < 0.2$. Beyond $\delta \geq 0.06$, we notice that $\langle F_{\min}^{2} \rangle_{\delta} < 0.5$ since the effects of localization are too strong to be counteracted by QEC. 

This is further borne out by our detailed analysis of the distribution of the transition amplitude in the presence of disorder (see Appendix~\ref{sec:transAmp_dist}). In particular, our expressions for the mean and standard deviation of the transition amplitude indicate that until $\delta \leq 0.01$, the disorder-averaged value $\langle f_{r,s}^{N}(t,\{\Delta_{k}\})\rangle_{\delta}$ is close to the value of the transition amplitude in the ideal (zero-disorder) case, and the standard deviation is insignificant compared to the mean. However, as the disorder strength increases further, the disorder-averaged value $\langle f_{r,s}^{N}(t,\{\Delta_{k}\})\rangle_{\delta}$ starts dropping and the standard deviation becomes comparable to the average value. Thus the effective noise parameter of the underlying quantum channel becomes too strong for the QEC procedure to be effective. 

The fact that $\delta = 0.06$ is a threshold of sorts can be seen more directly by studying the variation of the disorder-averaged transition amplitude with disorder strength. Previous studies on localization in disordered chains have used such a quantity, namely  $\langle |f_{n,1}^{N}(t,\{\Delta_{k}\})|^{2} \rangle_{\delta}$, as an indicator of the extent of localization~\cite{ashhab,Chiara}.

\begin{figure} [H]
\centering
\includegraphics[width=0.47\textwidth]{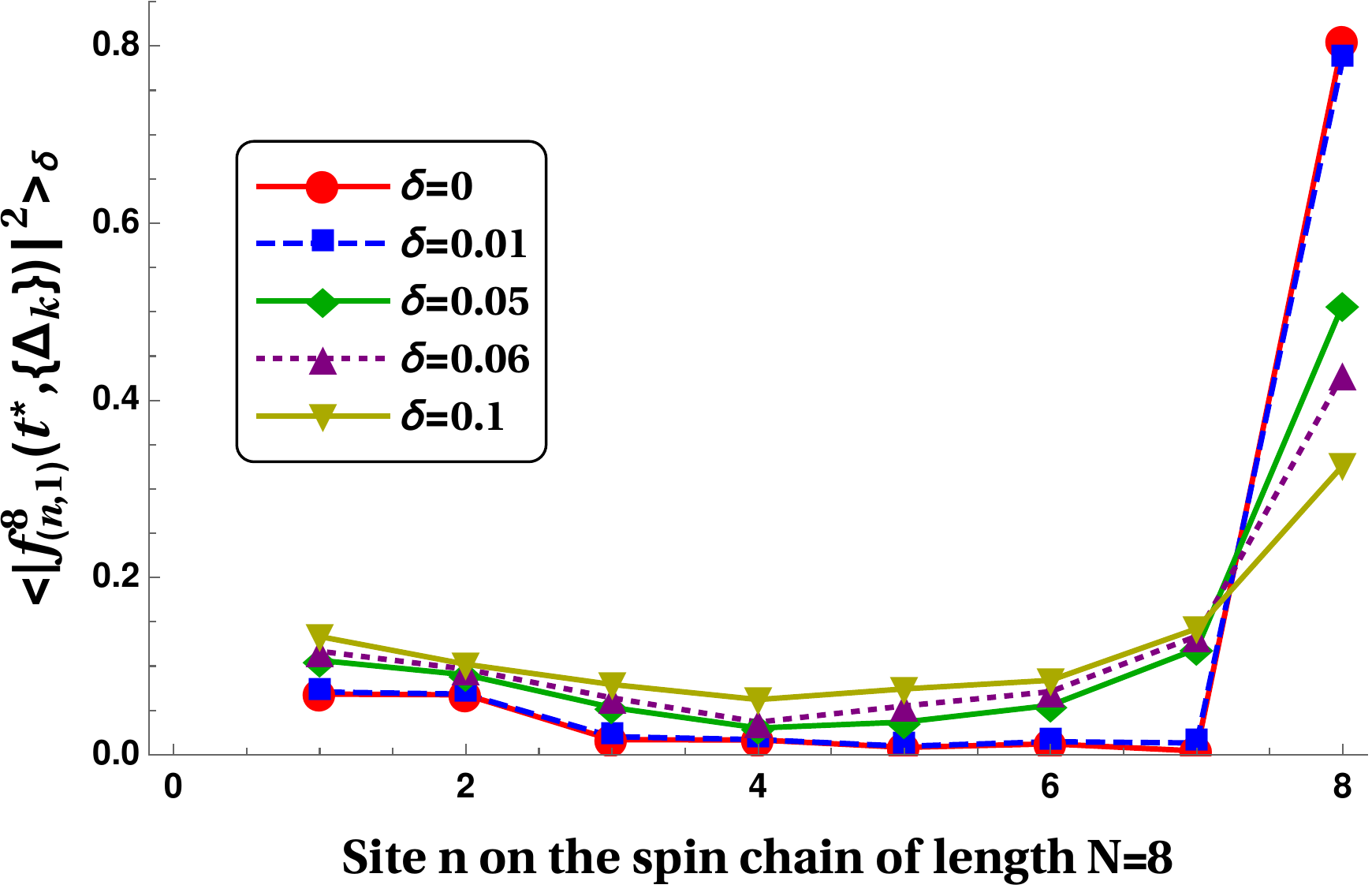}
\caption{$\langle |f_{n,1}^{8}(t)|^{2}\rangle_\delta$ for an $8$-spin chain  as a function of the site $n$.}
\label{fig:transAmp_dis8}
\end{figure}

In Fig.~\ref{fig:transAmp_dis8}, we plot the disorder-averaged transition amplitude $\langle |f_{n,1}^{N}(t,\{\Delta_{k}\})|^{2} \rangle_{\delta}$ for a fixed time $t^{*}$ and different disorder strengths $\delta$, as a function of the receiver site $n$, for the Heisenberg chain in Eq.~\eqref{eq:H_dis}. Empirically, we see that this plot follows an exponential distribution. The curves take the form $e^{-(\alpha n + \beta)/{\rm Loc}}$, where $\alpha, \beta$ are functions of disorder strength $\delta$ and ${\rm Loc}$ is the localization length, i.e. the length at which $\langle |f_{n,1}^{N}(t,\{\Delta_{k}\})|^2\rangle_{\delta}$ falls to $(1/e)$ of its maximum value. We see that with increase in disorder strength $\delta$, the localization effects become more pronounced. 

Specifically, when the disorder strength crosses $\delta=0.06$, the square of the transition amplitude between the ends of the $8$-spin chain falls below $0.43$ on average. However, we know from the fidelity estimate in Theorem~\ref{thm:aqec_fid} that the adaptive QEC protocol improves fidelity if only if $|f_{N,1}^{N}(t)|^{2} > 0.43$. Thus, for end to end state transfer on an $8$-length disordered Heisenberg chain, $\delta=0.06$ is indeed a threshold beyond which the adaptive QEC protocol cannot help in improving fidelity. Since our analysis of the distribution of the transition amplitude presented in Sec.~\ref{sec:transAmp_disorder} as well as the fidelity expression in Theorem~\ref{thm:aqec_fid} hold for any $s,r,N$ we can always identify such a threshold for a specific set of values. 

\section{Conclusions}\label{sec:concl}

We develop a pretty good state transfer protocol based on adaptive quantum error correction (QEC), for a universal class of Hamiltonians which preserve the total spin excitations on a linear spin chain. Based on the structure of the underlying quantum channel, we choose an approximate code and near-optimal, adaptive recovery map, to solve for the fidelity of state transfer explicitly. For the specific case of the ideal Heisenberg chain, our protocol performs as efficiently as perfect-QEC-based protocols. Using repeated QEC on the chain, we are able to achieve high enough fidelity over longer distances for an ideal spin chain. 

In the case of disordered spin chains the underlying quantum channel is stochastic. For the case of a disordered $1$-d Heisenberg chain, we study the distribution of the transition amplitude, which in turn is directly related to the stochastic noise parameter of the noise channel. By suitably adapting the recovery procedure, we demonstrate pretty good transfer on average, for low disorder strengths. 

It is an interesting question as to whether such channel-adapted QEC techniques maybe used to achieve pretty good state transfer for other universal classes, such as the transverse-field Ising model and the $XYZ$-chain. It is also an open problem to obtain an efficient circuit implementation of the adaptive recovery map discussed here.

\section{Acknowledgements}
The authors thank Hui Khoon Ng for insightful discussions. We also thank the anonymous referees for their useful comments and suggestions.

%


\begin{thebibliography}{31}%
\makeatletter
\providecommand \@ifxundefined [1]{%
 \@ifx{#1\undefined}
}%
\providecommand \@ifnum [1]{%
 \ifnum #1\expandafter \@firstoftwo
 \else \expandafter \@secondoftwo
 \fi
}%
\providecommand \@ifx [1]{%
 \ifx #1\expandafter \@firstoftwo
 \else \expandafter \@secondoftwo
 \fi
}%
\providecommand \natexlab [1]{#1}%
\providecommand \enquote  [1]{``#1''}%
\providecommand \bibnamefont  [1]{#1}%
\providecommand \bibfnamefont [1]{#1}%
\providecommand \citenamefont [1]{#1}%
\providecommand \href@noop [0]{\@secondoftwo}%
\providecommand \href [0]{\begingroup \@sanitize@url \@href}%
\providecommand \@href[1]{\@@startlink{#1}\@@href}%
\providecommand \@@href[1]{\endgroup#1\@@endlink}%
\providecommand \@sanitize@url [0]{\catcode `\\12\catcode `\$12\catcode
  `\&12\catcode `\#12\catcode `\^12\catcode `\_12\catcode `\%12\relax}%
\providecommand \@@startlink[1]{}%
\providecommand \@@endlink[0]{}%
\providecommand \url  [0]{\begingroup\@sanitize@url \@url }%
\providecommand \@url [1]{\endgroup\@href {#1}{\urlprefix }}%
\providecommand \urlprefix  [0]{URL }%
\providecommand \Eprint [0]{\href }%
\providecommand \doibase [0]{http://dx.doi.org/}%
\providecommand \selectlanguage [0]{\@gobble}%
\providecommand \bibinfo  [0]{\@secondoftwo}%
\providecommand \bibfield  [0]{\@secondoftwo}%
\providecommand \translation [1]{[#1]}%
\providecommand \BibitemOpen [0]{}%
\providecommand \bibitemStop [0]{}%
\providecommand \bibitemNoStop [0]{.\EOS\space}%
\providecommand \EOS [0]{\spacefactor3000\relax}%
\providecommand \BibitemShut  [1]{\csname bibitem#1\endcsname}%
\let\auto@bib@innerbib\@empty
\bibitem [{\citenamefont {Bose}(2003)}]{bose}%
  \BibitemOpen
  \bibfield  {author} {\bibinfo {author} {\bibfnamefont {S.}~\bibnamefont
  {Bose}},\ }\href@noop {} {\bibfield  {journal} {\bibinfo  {journal} {Physical
  review letters}\ }\textbf {\bibinfo {volume} {91}},\ \bibinfo {pages}
  {207901} (\bibinfo {year} {2003})}\BibitemShut {NoStop}%
\bibitem [{\citenamefont {Christandl}\ \emph {et~al.}(2004)\citenamefont
  {Christandl}, \citenamefont {Datta}, \citenamefont {Ekert},\ and\
  \citenamefont {Landahl}}]{christandl}%
  \BibitemOpen
  \bibfield  {author} {\bibinfo {author} {\bibfnamefont {M.}~\bibnamefont
  {Christandl}}, \bibinfo {author} {\bibfnamefont {N.}~\bibnamefont {Datta}},
  \bibinfo {author} {\bibfnamefont {A.}~\bibnamefont {Ekert}}, \ and\ \bibinfo
  {author} {\bibfnamefont {A.~J.}\ \bibnamefont {Landahl}},\ }\href@noop {}
  {\bibfield  {journal} {\bibinfo  {journal} {Physical review letters}\
  }\textbf {\bibinfo {volume} {92}},\ \bibinfo {pages} {187902} (\bibinfo
  {year} {2004})}\BibitemShut {NoStop}%
\bibitem [{\citenamefont {Christandl}\ \emph {et~al.}(2005)\citenamefont
  {Christandl}, \citenamefont {Datta}, \citenamefont {Dorlas}, \citenamefont
  {Ekert}, \citenamefont {Kay},\ and\ \citenamefont
  {Landahl}}]{christandl2005perfect}%
  \BibitemOpen
  \bibfield  {author} {\bibinfo {author} {\bibfnamefont {M.}~\bibnamefont
  {Christandl}}, \bibinfo {author} {\bibfnamefont {N.}~\bibnamefont {Datta}},
  \bibinfo {author} {\bibfnamefont {T.~C.}\ \bibnamefont {Dorlas}}, \bibinfo
  {author} {\bibfnamefont {A.}~\bibnamefont {Ekert}}, \bibinfo {author}
  {\bibfnamefont {A.}~\bibnamefont {Kay}}, \ and\ \bibinfo {author}
  {\bibfnamefont {A.~J.}\ \bibnamefont {Landahl}},\ }\href@noop {} {\bibfield
  {journal} {\bibinfo  {journal} {Physical Review A}\ }\textbf {\bibinfo
  {volume} {71}},\ \bibinfo {pages} {032312} (\bibinfo {year}
  {2005})}\BibitemShut {NoStop}%
\bibitem [{\citenamefont {Albanese}\ \emph {et~al.}(2004)\citenamefont
  {Albanese}, \citenamefont {Christandl}, \citenamefont {Datta},\ and\
  \citenamefont {Ekert}}]{albanesemirror}%
  \BibitemOpen
  \bibfield  {author} {\bibinfo {author} {\bibfnamefont {C.}~\bibnamefont
  {Albanese}}, \bibinfo {author} {\bibfnamefont {M.}~\bibnamefont
  {Christandl}}, \bibinfo {author} {\bibfnamefont {N.}~\bibnamefont {Datta}}, \
  and\ \bibinfo {author} {\bibfnamefont {A.}~\bibnamefont {Ekert}},\
  }\href@noop {} {\bibfield  {journal} {\bibinfo  {journal} {Physical review
  letters}\ }\textbf {\bibinfo {volume} {93}},\ \bibinfo {pages} {230502}
  (\bibinfo {year} {2004})}\BibitemShut {NoStop}%
\bibitem [{\citenamefont {Karbach}\ and\ \citenamefont
  {Stolze}(2005)}]{karbach}%
  \BibitemOpen
  \bibfield  {author} {\bibinfo {author} {\bibfnamefont {P.}~\bibnamefont
  {Karbach}}\ and\ \bibinfo {author} {\bibfnamefont {J.}~\bibnamefont
  {Stolze}},\ }\href@noop {} {\bibfield  {journal} {\bibinfo  {journal}
  {Physical Review A}\ }\textbf {\bibinfo {volume} {72}},\ \bibinfo {pages}
  {030301} (\bibinfo {year} {2005})}\BibitemShut {NoStop}%
\bibitem [{\citenamefont {Di~Franco}\ \emph {et~al.}(2008)\citenamefont
  {Di~Franco}, \citenamefont {Paternostro},\ and\ \citenamefont {Kim}}]{di}%
  \BibitemOpen
  \bibfield  {author} {\bibinfo {author} {\bibfnamefont {C.}~\bibnamefont
  {Di~Franco}}, \bibinfo {author} {\bibfnamefont {M.}~\bibnamefont
  {Paternostro}}, \ and\ \bibinfo {author} {\bibfnamefont {M.~S.}~\bibnamefont
  {Kim}},\ }\href@noop {} {\bibfield  {journal} {\bibinfo  {journal} {Physical
  review letters}\ }\textbf {\bibinfo {volume} {101}},\ \bibinfo {pages}
  {230502} (\bibinfo {year} {2008})}\BibitemShut {NoStop}%
\bibitem [{\citenamefont {Burgarth}\ and\ \citenamefont
  {Bose}(2005{\natexlab{a}})}]{conclusive}%
  \BibitemOpen
  \bibfield  {author} {\bibinfo {author} {\bibfnamefont {D.}~\bibnamefont
  {Burgarth}}\ and\ \bibinfo {author} {\bibfnamefont {S.}~\bibnamefont
  {Bose}},\ }\href@noop {} {\bibfield  {journal} {\bibinfo  {journal} {Physical
  Review A}\ }\textbf {\bibinfo {volume} {71}},\ \bibinfo {pages} {052315}
  (\bibinfo {year} {2005}{\natexlab{a}})}\BibitemShut {NoStop}%
\bibitem [{\citenamefont {Burgarth}\ and\ \citenamefont
  {Bose}(2005{\natexlab{b}})}]{perfect}%
  \BibitemOpen
  \bibfield  {author} {\bibinfo {author} {\bibfnamefont {D.}~\bibnamefont
  {Burgarth}}\ and\ \bibinfo {author} {\bibfnamefont {S.}~\bibnamefont
  {Bose}},\ }\href@noop {} {\bibfield  {journal} {\bibinfo  {journal} {New
  journal of physics}\ }\textbf {\bibinfo {volume} {7}},\ \bibinfo {pages}
  {135} (\bibinfo {year} {2005}{\natexlab{b}})}\BibitemShut {NoStop}%
\bibitem [{\citenamefont {Burgarth}\ \emph {et~al.}(2005)\citenamefont
  {Burgarth}, \citenamefont {Giovannetti},\ and\ \citenamefont
  {Bose}}]{efficient}%
  \BibitemOpen
  \bibfield  {author} {\bibinfo {author} {\bibfnamefont {D.}~\bibnamefont
  {Burgarth}}, \bibinfo {author} {\bibfnamefont {V.}~\bibnamefont
  {Giovannetti}}, \ and\ \bibinfo {author} {\bibfnamefont {S.}~\bibnamefont
  {Bose}},\ }\href@noop {} {\bibfield  {journal} {\bibinfo  {journal} {Journal
  of Physics A: Mathematical and General}\ }\textbf {\bibinfo {volume} {38}},\
  \bibinfo {pages} {6793} (\bibinfo {year} {2005})}\BibitemShut {NoStop}%
\bibitem [{\citenamefont {Bochkin}\ \emph {et~al.}(2016)\citenamefont
  {Bochkin}, \citenamefont {Doronin}, \citenamefont {Vasil'ev}, \citenamefont
  {Fedorova},\ and\ \citenamefont {Fel'dman}}]{bochkin}%
  \BibitemOpen
  \bibfield  {author} {\bibinfo {author} {\bibfnamefont {G.}~\bibnamefont
  {Bochkin}}, \bibinfo {author} {\bibfnamefont {S.}~\bibnamefont {Doronin}},
  \bibinfo {author} {\bibfnamefont {S.}~\bibnamefont {Vasil'ev}}, \bibinfo
  {author} {\bibfnamefont {A.}~\bibnamefont {Fedorova}}, \ and\ \bibinfo
  {author} {\bibfnamefont {E.}~\bibnamefont {Fel'dman}},\ }in\ \href@noop {}
  {\emph {\bibinfo {booktitle} {International Conference on Micro-and
  Nano-Electronics 2016}}},\ Vol.\ \bibinfo {volume} {10224},\, \bibinfo{edit}{edited by V. F. Lukichev and K. V. Rudenko} (\bibinfo
  {organization} {International Society for Optics and Photonics},\ \bibinfo
  {year} {2016})\ p.\ \bibinfo {pages} {102242E}\BibitemShut {NoStop}%
\bibitem [{\citenamefont {Perez-Leija}\ \emph {et~al.}(2013)\citenamefont
  {Perez-Leija}, \citenamefont {Keil}, \citenamefont {Kay}, \citenamefont
  {Moya-Cessa}, \citenamefont {Nolte}, \citenamefont {Kwek}, \citenamefont
  {Rodr{\'\i}guez-Lara}, \citenamefont {Szameit},\ and\ \citenamefont
  {Christodoulides}}]{perez2013}%
  \BibitemOpen
  \bibfield  {author} {\bibinfo {author} {\bibfnamefont {A.}~\bibnamefont
  {Perez-Leija}}, \bibinfo {author} {\bibfnamefont {R.}~\bibnamefont {Keil}},
  \bibinfo {author} {\bibfnamefont {A.}~\bibnamefont {Kay}}, \bibinfo {author}
  {\bibfnamefont {H.}~\bibnamefont {Moya-Cessa}}, \bibinfo {author}
  {\bibfnamefont {S.}~\bibnamefont {Nolte}}, \bibinfo {author} {\bibfnamefont
  {L.-C.}\ \bibnamefont {Kwek}}, \bibinfo {author} {\bibfnamefont {B.~M.}\
  \bibnamefont {Rodr{\'\i}guez-Lara}}, \bibinfo {author} {\bibfnamefont
  {A.}~\bibnamefont {Szameit}}, \ and\ \bibinfo {author} {\bibfnamefont
  {D.~N.}\ \bibnamefont {Christodoulides}},\ }\href@noop {} {\bibfield
  {journal} {\bibinfo  {journal} {Physical Review A}\ }\textbf {\bibinfo
  {volume} {87}},\ \bibinfo {pages} {012309} (\bibinfo {year}
  {2013})}\BibitemShut {NoStop}%
\bibitem [{\citenamefont {Chapman}\ \emph {et~al.}(2016)\citenamefont
  {Chapman}, \citenamefont {Santandrea}, \citenamefont {Huang}, \citenamefont
  {Corrielli}, \citenamefont {Crespi}, \citenamefont {Yung}, \citenamefont
  {Osellame},\ and\ \citenamefont {Peruzzo}}]{chapman}%
  \BibitemOpen
  \bibfield  {author} {\bibinfo {author} {\bibfnamefont {R.~J.}\ \bibnamefont
  {Chapman}}, \bibinfo {author} {\bibfnamefont {M.}~\bibnamefont {Santandrea}},
  \bibinfo {author} {\bibfnamefont {Z.}~\bibnamefont {Huang}}, \bibinfo
  {author} {\bibfnamefont {G.}~\bibnamefont {Corrielli}}, \bibinfo {author}
  {\bibfnamefont {A.}~\bibnamefont {Crespi}}, \bibinfo {author} {\bibfnamefont
  {M.-H.}\ \bibnamefont {Yung}}, \bibinfo {author} {\bibfnamefont
  {R.}~\bibnamefont {Osellame}}, \ and\ \bibinfo {author} {\bibfnamefont
  {A.}~\bibnamefont {Peruzzo}},\ }\href@noop {} {\bibfield  {journal} {\bibinfo
   {journal} {Nature communications}\ }\textbf {\bibinfo {volume} {7}},\
  \bibinfo {pages} {11339} (\bibinfo {year} {2016})}\BibitemShut {NoStop}%
\bibitem [{\citenamefont {Godsil}\ \emph {et~al.}(2012)\citenamefont {Godsil},
  \citenamefont {Kirkland}, \citenamefont {Severini},\ and\ \citenamefont
  {Smith}}]{godsil2012}%
  \BibitemOpen
  \bibfield  {author} {\bibinfo {author} {\bibfnamefont {C.}~\bibnamefont
  {Godsil}}, \bibinfo {author} {\bibfnamefont {S.}~\bibnamefont {Kirkland}},
  \bibinfo {author} {\bibfnamefont {S.}~\bibnamefont {Severini}}, \ and\
  \bibinfo {author} {\bibfnamefont {J.}~\bibnamefont {Smith}},\ }\href@noop {}
  {\bibfield  {journal} {\bibinfo  {journal} {Physical review letters}\
  }\textbf {\bibinfo {volume} {109}},\ \bibinfo {pages} {050502} (\bibinfo
  {year} {2012})}\BibitemShut {NoStop}%
\bibitem [{\citenamefont {Banchi}\ \emph {et~al.}(2017)\citenamefont {Banchi},
  \citenamefont {Coutinho}, \citenamefont {Godsil},\ and\ \citenamefont
  {Severini}}]{godsil}%
  \BibitemOpen
  \bibfield  {author} {\bibinfo {author} {\bibfnamefont {L.}~\bibnamefont
  {Banchi}}, \bibinfo {author} {\bibfnamefont {G.}~\bibnamefont {Coutinho}},
  \bibinfo {author} {\bibfnamefont {C.}~\bibnamefont {Godsil}}, \ and\ \bibinfo
  {author} {\bibfnamefont {S.}~\bibnamefont {Severini}},\ }\href@noop {}
  {\bibfield  {journal} {\bibinfo  {journal} {Journal of Mathematical Physics}\
  }\textbf {\bibinfo {volume} {58}},\ \bibinfo {pages} {032202} (\bibinfo
  {year} {2017})}\BibitemShut {NoStop}%
\bibitem [{\citenamefont {Osborne}\ and\ \citenamefont
  {Linden}(2004)}]{osborne}%
  \BibitemOpen
  \bibfield  {author} {\bibinfo {author} {\bibfnamefont {T.~J.}\ \bibnamefont
  {Osborne}}\ and\ \bibinfo {author} {\bibfnamefont {N.}~\bibnamefont
  {Linden}},\ }\href@noop {} {\bibfield  {journal} {\bibinfo  {journal}
  {Physical Review A}\ }\textbf {\bibinfo {volume} {69}},\ \bibinfo {pages}
  {052315} (\bibinfo {year} {2004})}\BibitemShut {NoStop}%
\bibitem [{\citenamefont {Haselgrove}(2005)}]{hasel}%
  \BibitemOpen
  \bibfield  {author} {\bibinfo {author} {\bibfnamefont {H.~L.}\ \bibnamefont
  {Haselgrove}},\ }\href@noop {} {\bibfield  {journal} {\bibinfo  {journal}
  {Physical Review A}\ }\textbf {\bibinfo {volume} {72}},\ \bibinfo {pages}
  {062326} (\bibinfo {year} {2005})}\BibitemShut {NoStop}%
\bibitem [{\citenamefont {Ashhab}(2015)}]{ashhab}%
  \BibitemOpen
  \bibfield  {author} {\bibinfo {author} {\bibfnamefont {S.}~\bibnamefont
  {Ashhab}},\ }\href@noop {} {\bibfield  {journal} {\bibinfo  {journal}
  {Physical Review A}\ }\textbf {\bibinfo {volume} {92}},\ \bibinfo {pages}
  {062305} (\bibinfo {year} {2015})}\BibitemShut {NoStop}%
\bibitem [{\citenamefont {De~Chiara}\ \emph {et~al.}(2005)\citenamefont
  {De~Chiara}, \citenamefont {Rossini}, \citenamefont {Montangero},\ and\
  \citenamefont {Fazio}}]{Chiara}%
  \BibitemOpen
  \bibfield  {author} {\bibinfo {author} {\bibfnamefont {G.}~\bibnamefont
  {De~Chiara}}, \bibinfo {author} {\bibfnamefont {D.}~\bibnamefont {Rossini}},
  \bibinfo {author} {\bibfnamefont {S.}~\bibnamefont {Montangero}}, \ and\
  \bibinfo {author} {\bibfnamefont {R.}~\bibnamefont {Fazio}},\ }\href@noop {}
  {\bibfield  {journal} {\bibinfo  {journal} {Physical Review A}\ }\textbf
  {\bibinfo {volume} {72}},\ \bibinfo {pages} {012323} (\bibinfo {year}
  {2005})}\BibitemShut {NoStop}%
\bibitem [{\citenamefont {Nielsen}\ and\ \citenamefont
  {Chuang}(2000)}]{Nielsen}%
  \BibitemOpen
  \bibfield  {author} {\bibinfo {author} {\bibfnamefont {M.}~\bibnamefont
  {Nielsen}}\ and\ \bibinfo {author} {\bibfnamefont {I.}~\bibnamefont
  {Chuang}},\ }\href@noop {} {\emph {\bibinfo {title} {Quantum Computation and
  Quantum Information}}}\ (\bibinfo  {publisher} {Cambridge University Press, Cambridge},\
  \bibinfo {year} {2000})\BibitemShut {NoStop}%
\bibitem [{\citenamefont {Kay}(2016)}]{kay}%
  \BibitemOpen
  \bibfield  {author} {\bibinfo {author} {\bibfnamefont {A.}~\bibnamefont
  {Kay}},\ }\href@noop {} {\bibfield  {journal} {\bibinfo  {journal} {Physical
  Review A}\ }\textbf {\bibinfo {volume} {93}},\ \bibinfo {pages} {042320}
  (\bibinfo {year} {2016})}\BibitemShut {NoStop}%
\bibitem [{\citenamefont {Kay}(2018)}]{kay2018perfect}%
  \BibitemOpen
  \bibfield  {author} {\bibinfo {author} {\bibfnamefont {A.}~\bibnamefont
  {Kay}},\ }\href@noop {} {\bibfield  {journal} {\bibinfo  {journal} {Physical
  Review A}\ }\textbf {\bibinfo {volume} {97}},\ \bibinfo {pages} {032317}
  (\bibinfo {year} {2018})}\BibitemShut {NoStop}%
\bibitem [{\citenamefont {Allcock}\ and\ \citenamefont
  {Linden}(2009)}]{allcock}%
  \BibitemOpen
  \bibfield  {author} {\bibinfo {author} {\bibfnamefont {J.}~\bibnamefont
  {Allcock}}\ and\ \bibinfo {author} {\bibfnamefont {N.}~\bibnamefont
  {Linden}},\ }\href@noop {} {\bibfield  {journal} {\bibinfo  {journal}
  {Physical review letters}\ }\textbf {\bibinfo {volume} {102}},\ \bibinfo
  {pages} {110501} (\bibinfo {year} {2009})}\BibitemShut {NoStop}%
\bibitem [{\citenamefont {Leung}\ \emph {et~al.}(1997)\citenamefont {Leung},
  \citenamefont {Nielsen}, \citenamefont {Chuang},\ and\ \citenamefont
  {Yamamoto}}]{Leung}%
  \BibitemOpen
  \bibfield  {author} {\bibinfo {author} {\bibfnamefont {D.~W.}\ \bibnamefont
  {Leung}}, \bibinfo {author} {\bibfnamefont {M.~A.}\ \bibnamefont {Nielsen}},
  \bibinfo {author} {\bibfnamefont {I.~L.}\ \bibnamefont {Chuang}}, \ and\
  \bibinfo {author} {\bibfnamefont {Y.}~\bibnamefont {Yamamoto}},\ }\href@noop
  {} {\bibfield  {journal} {\bibinfo  {journal} {Phys. Rev. A}\ }\textbf
  {\bibinfo {volume} {56}},\ \bibinfo {pages} {2567} (\bibinfo {year}
  {1997})}\BibitemShut {NoStop}%
\bibitem [{\citenamefont {Fletcher}\ \emph {et~al.}(2008)\citenamefont
  {Fletcher}, \citenamefont {Shor},\ and\ \citenamefont {Win}}]{Fletcher}%
  \BibitemOpen
  \bibfield  {author} {\bibinfo {author} {\bibfnamefont {A.~S.}\ \bibnamefont
  {Fletcher}}, \bibinfo {author} {\bibfnamefont {P.~W.}\ \bibnamefont {Shor}},
  \ and\ \bibinfo {author} {\bibfnamefont {M.~Z.}\ \bibnamefont {Win}},\
  }\href@noop {} {\bibfield  {journal} {\bibinfo  {journal} {IEEE Trans. Info.
  Theory}\ }\textbf {\bibinfo {volume} {54}},\ \bibinfo {pages} {5705}
  (\bibinfo {year} {2008})}\BibitemShut {NoStop}%
\bibitem [{\citenamefont {Ng}\ and\ \citenamefont
  {Mandayam}(2010)}]{HKN_PM2010}%
  \BibitemOpen
  \bibfield  {author} {\bibinfo {author} {\bibfnamefont {H.~K.}\ \bibnamefont
  {Ng}}\ and\ \bibinfo {author} {\bibfnamefont {P.}~\bibnamefont {Mandayam}},\
  }\href@noop {} {\bibfield  {journal} {\bibinfo  {journal} {Physical Review
  A}\ }\textbf {\bibinfo {volume} {81}},\ \bibinfo {pages} {062342} (\bibinfo
  {year} {2010})}\BibitemShut {NoStop}%
\bibitem [{\citenamefont {Cafaro}\ and\ \citenamefont {van
  Loock}(2014{\natexlab{a}})}]{cafaro}%
  \BibitemOpen
  \bibfield  {author} {\bibinfo {author} {\bibfnamefont {C.}~\bibnamefont
  {Cafaro}}\ and\ \bibinfo {author} {\bibfnamefont {P.}~\bibnamefont {van
  Loock}},\ }\href@noop {} {\bibfield  {journal} {\bibinfo  {journal} {Physical
  Review A}\ }\textbf {\bibinfo {volume} {89}},\ \bibinfo {pages} {022316}
  (\bibinfo {year} {2014}{\natexlab{a}})}\BibitemShut {NoStop}%
\bibitem [{\citenamefont {Anderson}(1958)}]{anderson}%
  \BibitemOpen
  \bibfield  {author} {\bibinfo {author} {\bibfnamefont {P.~W.}\ \bibnamefont
  {Anderson}},\ }\href@noop {} {\bibfield  {journal} {\bibinfo  {journal}
  {Physical review}\ }\textbf {\bibinfo {volume} {109}},\ \bibinfo {pages}
  {1492} (\bibinfo {year} {1958})}\BibitemShut {NoStop}%
\bibitem [{\citenamefont {Kay}(2010)}]{kayreview}%
  \BibitemOpen
  \bibfield  {author} {\bibinfo {author} {\bibfnamefont {A.}~\bibnamefont
  {Kay}},\ }\href@noop {} {\bibfield  {journal} {\bibinfo  {journal}
  {International Journal of Quantum Information}\ }\textbf {\bibinfo {volume}
  {8}},\ \bibinfo {pages} {641} (\bibinfo {year} {2010})}\BibitemShut {NoStop}%
\bibitem [{\citenamefont {Piedrafita}\ and\ \citenamefont
  {Renes}(2017)}]{AD_reliable2017}%
  \BibitemOpen
  \bibfield  {author} {\bibinfo {author} {\bibfnamefont {{\'A}.}~\bibnamefont
  {Piedrafita}}\ and\ \bibinfo {author} {\bibfnamefont {J.~M.}\ \bibnamefont
  {Renes}},\ }\href@noop {} {\bibfield  {journal} {\bibinfo  {journal}
  {Physical review letters}\ }\textbf {\bibinfo {volume} {119}},\ \bibinfo
  {pages} {250501} (\bibinfo {year} {2017})}\BibitemShut {NoStop}%
\bibitem [{\citenamefont {Cafaro}\ and\ \citenamefont {van
  Loock}(2014{\natexlab{b}})}]{cafaro2014simple}%
  \BibitemOpen
  \bibfield  {author} {\bibinfo {author} {\bibfnamefont {C.}~\bibnamefont
  {Cafaro}}\ and\ \bibinfo {author} {\bibfnamefont {P.}~\bibnamefont {van
  Loock}},\ }\href@noop {} {\bibfield  {journal} {\bibinfo  {journal} {Open
  Systems \& Information Dynamics}\ }\textbf {\bibinfo {volume} {21}},\
  \bibinfo {pages} {1450002} (\bibinfo {year}
  {2014}{\natexlab{b}})}\BibitemShut {NoStop}%
\bibitem [{\citenamefont {Dyson}(1949)}]{dyson}%
  \BibitemOpen
  \bibfield  {author} {\bibinfo {author} {\bibfnamefont {F.~J.}\ \bibnamefont
  {Dyson}},\ }\href@noop {} {\bibfield  {journal} {\bibinfo  {journal}
  {Physical Review}\ }\textbf {\bibinfo {volume} {75}},\ \bibinfo {pages} {486}
  (\bibinfo {year} {1949})}\BibitemShut {NoStop}%
  
\end{thebibliography}

\appendix
\begin{widetext}
\section{Effect of noise channel $\cE$ on $4$-qubit code}\label{sec:E(P)}

We note the following structure for the Kraus operators of the $4$-qubit channel, by expanding them in the $4$-qubit computational basis. First, we note that the only Kraus operator diagonal in the computational basis is $E_{0}^{\otimes 4}$, with diagonal entry $e^{ij\Theta}|f_{r,s}^{N}(t)|^{j}$, corresponding to those basis vectors with $j$ $1$'s in them. 
All the other operators are off-diagonal matrices with support on some subset of computational basis states. For example, a three-qubit error operator (involving $E_{1}$ in three of the four  qubits) is of the form,
\begin{equation}
E_{0}\otimes E_{1}^{\otimes 3} = (1-|f^{N}_{r,s}(t)|^{2})^{3/2}|0000\rangle\langle 0111|+ e^{i\Theta}|f^{N}_{r,s}(t)|(1-|f^{N}_{r,s}(t)|^{2})^{3/2}|1000\rangle\langle 1111| . \label{eq:E_03}
\end{equation}
The remaining three-qubit errors are of the same form, with the strings $\{0111, 1000\}$ replaced by their permutations. Similarly, an operator which has $E_{1}$ errors on two of the qubits is a linear combination of the form,
\begin{eqnarray}
E_{0}^{\otimes 2}\otimes E_{1}^{\otimes 2} &=& (1-|f^{N}_{r,s}(t)|^{2})|0000\rangle\langle 0011| + e^{2i\Theta}|f^{N}_{r,s}(t)|^{2}(1-|f^{N}_{r,s}(t)|^{2})|1100\rangle\langle 1111| \nonumber \\&& +  e^{i\Theta}|f^{N}_{r,s}(t)|(1-|f^{N}_{r,s}(t)|^{2})\left(|0100\rangle\langle 0111| + |1000\rangle\langle 1011| \right). \label{eq:E_02}
\end{eqnarray}
Other two-qubit error operators are realized by replacing the strings $\{0011,1100,0100,1000\}$ with permutations thereof. A single-qubit error operator, with $E_{1}$ error on only one of the qubits has the form,
\begin{eqnarray}
 E_{0}^{\otimes 3}\otimes E_{1} &=& \sqrt{1-|f^{N}_{r,s}(t)|^{2}}|0000\rangle \langle 0001| + e^{i\Theta}|f^{N}_{r,s}(t)|\sqrt{1-|f^{N}_{r,s}(t)|^{2}}\left(|0010\rangle\langle 0011| + |0100\rangle\langle 0101| + |1000\rangle\langle1001| \right) \nonumber \\ 
&+& e^{2i\Theta} |f^{N}_{r,s}(t)|^{2}\sqrt{1-|f^{N}_{r,s}(t)|^{2}}\left(|1100\rangle\langle 1101| + |0110\rangle\langle 0111| + |1010\rangle\langle 1011| \right)  \nonumber \\ && + e^{3i\Theta} |f^{N}_{r,s}(t)|^{3}\sqrt{1-|f^{N}_{r,s}(t)|^{2}}|1110\rangle\langle 1111| .  \label{eq:E_01}
\end{eqnarray}
Finally, the four-qubit error operator $E_{1}^{\otimes 4}$ is of the form,
\begin{equation}
E_{1}^{\otimes 4} = (1-|f^{N}_{r,s}(t)|^{2})^{2} |0000\rangle \langle 1111| . \label{eq:E_14}
\end{equation}

We next explicitly write out the operator $\small{\cE^{\otimes 4}(P)}$ in the computational basis of the $4$-qubit space.
\begin{equation}
\small{\cE^{\otimes 4}(P)} = 
 \left[
\begin{array}{*{16}c}
 \cQ_{1} & 0 & 0 & 0 & 0 & 0 & 0 & 0 & 0 & 0 & 0 & 0 & 0 & 0 & 0 &  e^{-4 i \Theta }\, \cQ_{17} \\
 0 & \cQ_{2} & 0 & 0 & 0 & 0 & 0 & 0 & 0 & 0 & 0 & 0 & 0 & 0 & 0 & 0 \\
 0 & 0 &\cQ_{3} & 0 & 0 & 0 & 0 & 0 & 0 & 0 & 0 & 0 & 0 & 0 & 0 & 0 \\
 0 & 0 & 0 & \cQ_{4}& 0 & 0 & 0 & 0 & 0 & 0 & 0 & 0 & \cQ_{18}& 0 & 0 & 0 \\
 0 & 0 & 0 & 0 & \cQ_{5} & 0 & 0 & 0 & 0 & 0 & 0 & 0 & 0 & 0 & 0 & 0 \\
 0 & 0 & 0 & 0 & 0 & \cQ_{6} & 0 & 0 & 0 & 0 & 0 & 0 & 0 & 0 & 0 & 0 \\
 0 & 0 & 0 & 0 & 0 & 0 & \cQ_{7}& 0 & 0 & 0 & 0 & 0 & 0 & 0 & 0 & 0 \\
 0 & 0 & 0 & 0 & 0 & 0 & 0 & \cQ_{8} & 0 & 0 & 0 & 0 & 0 & 0 & 0 & 0 \\
 0 & 0 & 0 & 0 & 0 & 0 & 0 & 0 & \cQ_{9} & 0 & 0 & 0 & 0 & 0 & 0 & 0 \\
 0 & 0 & 0 & 0 & 0 & 0 & 0 & 0 & 0 &\cQ_{10}& 0 & 0 & 0 & 0 & 0 & 0 \\
 0 & 0 & 0 & 0 & 0 & 0 & 0 & 0 & 0 & 0 &\cQ_{11} & 0 & 0 & 0 & 0 & 0 \\
 0 & 0 & 0 & 0 & 0 & 0 & 0 & 0 & 0 & 0 & 0 & \cQ_{12} & 0 & 0 & 0 & 0 \\
 0 & 0 & 0 & \cQ_{18} & 0 & 0 & 0 & 0 & 0 & 0 & 0 & 0 & \cQ_{13} & 0 & 0 & 0 \\
 0 & 0 & 0 & 0 & 0 & 0 & 0 & 0 & 0 & 0 & 0 & 0 & 0 & \cQ_{14}& 0 & 0 \\
 0 & 0 & 0 & 0 & 0 & 0 & 0 & 0 & 0 & 0 & 0 & 0 & 0 & 0 & \cQ_{15} & 0 \\
 e^{4 i \Theta }\,\cQ_{17}& 0 & 0 & 0 & 0 & 0 & 0 & 0 & 0 & 0 & 0 & 0 & 0 & 0 & 0 & \cQ_{16} \\
\end{array}
\right], \nonumber
\end{equation}
with $\{\cQ_{i}\}$ denoting polynomial functions of the transition amplitude $|f^{N}_{r,s}(t)|$. In terms of the rank-$1$ projectors onto the computational basis states, we may write $\small{\cE^{\otimes 4}(P)}$ as,
\begin{equation}
 \small{\cE^{\otimes 4}(P)}= \sum_{i=1}^{16}\cQ_{i}|i\rangle\langle i| + e^{-4 i \Theta } \cQ_{17}|0000\rangle\langle 1111| + e^{i 4\Theta} \cQ_{17}|1111\rangle\langle0000| + \cQ_{18}(|1100\rangle \langle0011| +|0011\rangle \langle1100|),
\end{equation}
wherein $|i\rangle \in \{|0000\rangle,\ldots,|0100\rangle, \ldots, |1111\rangle\}$ denote the computational basis states of the $4$-qubit space.

Similarly, we can also express the pseudo-inverse $\small{\cE^{\otimes 4}(P)^{-1/2}}$ in the $4$-qubit computational basis, as follows:
\begin{equation}
 \small{\cE^{\otimes 4}(P)^{-1/2}} = 
 \left[
\begin{array}{*{16}c}
 \cG_{1} & 0 & 0 & 0 & 0 & 0 & 0 & 0 & 0 & 0 & 0 & 0 & 0 & 0 & 0 &  e^{-4 i \Theta }\, \cG_{17} \\
 0 & \cG_{2} & 0 & 0 & 0 & 0 & 0 & 0 & 0 & 0 & 0 & 0 & 0 & 0 & 0 & 0 \\
 0 & 0 &\cG_{3} & 0 & 0 & 0 & 0 & 0 & 0 & 0 & 0 & 0 & 0 & 0 & 0 & 0 \\
 0 & 0 & 0 & \cG_{4}& 0 & 0 & 0 & 0 & 0 & 0 & 0 & 0 & \cG_{18}& 0 & 0 & 0 \\
 0 & 0 & 0 & 0 & \cG_{5} & 0 & 0 & 0 & 0 & 0 & 0 & 0 & 0 & 0 & 0 & 0 \\
 0 & 0 & 0 & 0 & 0 & \cG_{6} & 0 & 0 & 0 & 0 & 0 & 0 & 0 & 0 & 0 & 0 \\
 0 & 0 & 0 & 0 & 0 & 0 & \cG_{7}& 0 & 0 & 0 & 0 & 0 & 0 & 0 & 0 & 0 \\
 0 & 0 & 0 & 0 & 0 & 0 & 0 & \cG_{8} & 0 & 0 & 0 & 0 & 0 & 0 & 0 & 0 \\
 0 & 0 & 0 & 0 & 0 & 0 & 0 & 0 & \cG_{9} & 0 & 0 & 0 & 0 & 0 & 0 & 0 \\
 0 & 0 & 0 & 0 & 0 & 0 & 0 & 0 & 0 &\cG_{10}& 0 & 0 & 0 & 0 & 0 & 0 \\
 0 & 0 & 0 & 0 & 0 & 0 & 0 & 0 & 0 & 0 &\cG_{11} & 0 & 0 & 0 & 0 & 0 \\
 0 & 0 & 0 & 0 & 0 & 0 & 0 & 0 & 0 & 0 & 0 & \cG_{12} & 0 & 0 & 0 & 0 \\
 0 & 0 & 0 & \cG_{18} & 0 & 0 & 0 & 0 & 0 & 0 & 0 & 0 & \cG_{13} & 0 & 0 & 0 \\
 0 & 0 & 0 & 0 & 0 & 0 & 0 & 0 & 0 & 0 & 0 & 0 & 0 & \cG_{14}& 0 & 0 \\
 0 & 0 & 0 & 0 & 0 & 0 & 0 & 0 & 0 & 0 & 0 & 0 & 0 & 0 & \cG_{15} & 0 \\
 e^{4 i \Theta }\,\cG_{17}& 0 & 0 & 0 & 0 & 0 & 0 & 0 & 0 & 0 & 0 & 0 & 0 & 0 & 0 & \cG_{16} \\
\end{array}
\right], 
\end{equation}
with $\{\cG_{i}\}$ denoting a set of polynomials in $|f^{N}_{r,s}(t)|$. In terms of the rank-$1$ projectors onto the computational basis states, we have,
\begin{equation}
 \small{\cE^{\otimes 4}(P)^{-1/2}}= \sum_{i=1}^{16}\cG_{i}|i\rangle\langle i| + e^{-4 i \Theta } \cG_{17}|0000\rangle\langle1111| + e^{i 4\Theta} \cG_{17}|1111\rangle\langle0000| + \cG_{18}(|1100\rangle \langle0011| +|0011\rangle \langle1100|). \label{eq:EP_inv}
\end{equation}
Upon sandwiching the operator in Eq.~\eqref{eq:EP_inv} between the different error operators of the four-qubit noise channel (as described in Eqs.~\eqref{eq:E_03},~\eqref{eq:E_02},~\eqref{eq:E_01},~\eqref{eq:E_14}) and their adjoints, it is easy to see that the phases cancel out everywhere. In other words, the Kraus operators of the composite channel comprising noise and recovery are all independent of the phase $\Theta$ of the transition amplitude.

\section{Distribution of the transition amplitude for a disordered $XXX$ chain}\label{sec:transAmp_dist}

Here we derive the distribution of the transition amplitude $f^{N}_{r,s}(t,\{\Delta_{k}\})$ for the disordered $XXX$ chain described in Eq.~\eqref{eq:H_dis}, as a function of time $t$ and disorder strength $\delta$. Recall that the transition amplitude between the $r^{\rm th}$ and $s^{\rm th}$ site for the disordered Hamiltonian $\mathcal{H}_{\rm dis}$ is given by,
\begin{equation}
f^{N}_{r,s}(t,\{\Delta_{k}\}) = \langle \textbf{r} | e^{- i (\mathcal{H}_{o}+ \mathcal{H}_{\delta})t} |\textbf{s}\rangle = \langle \textbf{r}| e^{-i \mathcal{H}_{o} t}\mathcal{ T}\left[\exp{\left(-i\int_{0}^{t}e^{i\mathcal{ H}_{o} t'}\,\mathcal{H_{\delta}}\,e^{-i \mathcal{H}_{o} t'}dt'\right)}\right]| \textbf{s}\rangle, \label{eq:transAmp_delta}
\end{equation}
where $\cT$ denotes the time-ordering operator. We first expand the time-ordered perturbation series in Eq.~\eqref{eq:transAmp_delta} as follows, 
\begin{eqnarray}
f^{N}_{r,s}(t,\{\Delta_{k}\}) &=&\sum_{k=1}^{N}\langle \textbf{r}| e^{-i \mathcal{H}_{o} t} |\textbf{k}\rangle \langle \textbf{k}|\mathcal{T}\left[e^{(-i\int_{0}^{t}e^{i\mathcal{ H}_{o} t'}\mathcal{H_{\delta}}e^{-i \mathcal{H}_{o} t'}dt')}\right] |\textbf{s}\rangle  \nonumber \\
&=& \sum_{k=1}^{N}f^{N}_{r,k}(t) \langle \textbf{k} | \, I - i O(\mathcal{H}_{\delta}) + \frac{i^{2}}{2!} O(\mathcal{H}_{\delta}^{2}) + \ldots \, | \textbf{s}\rangle \label{eq:transAmp_disorder}
\end{eqnarray}
where, $f^{N}_{r,k}(t) = \langle \textbf{r}| e^{-i\mathcal{ H}_{o} t} |\textbf{k}\rangle$ is the transition amplitude in the absence of disorder. Expanding the first order term ($O(\mathcal{H}_{\delta})$) as a time-ordered form, we have,
\begin{eqnarray}
\langle \textbf {k} |O(\mathcal{H}_{\delta}) |\textbf{s} \rangle &=& \int_{0}^{t}\langle \textbf{k}|e^{i \mathcal{H}_{o} t'} \mathcal{H}_{\delta}e^{-i\mathcal{ H}_{o} t'}|\textbf{s}\rangle dt' \nonumber \\
&=& \sum_{l,m=1}^{N} \int_{0}^{t}\langle \textbf{k}|e^{i \mathcal{H}_{o} t'} |\textbf{l}\rangle \langle \textbf{l}|\mathcal{H}_{\delta}|\textbf{m}\rangle\langle \textbf{m}|e^{-i\mathcal{ H}_{o} t'}|\textbf{s}\rangle dt' \label{eq:first_order1}
\end{eqnarray}
where,
\begin{equation}
\langle \textbf{l} | \mathcal{H}_{\delta} |\textbf{m}\rangle = \frac{\overline{J}}{2}\left(\sum_{i=1}^{N-1}( u_{i}^l \Delta_{i})\delta_{lm} - 2\Delta_{l}\delta_{m(l+1)}  - 2\Delta_{l-1}\delta_{m(l-1)} \right), \label{eq:first_order2}
\end{equation}
with the coefficients $u^{l}_{i} \in \{\pm 1\}$. For example, $\cH_{\delta}$ for a $4$-qubit spin chain is a tridiagonal matrix of the form,

\begin{equation}
\cH_{\delta} =
\frac{\overline{J}}{2}\left(
\begin{array}{cccc}
 -\Delta_{1}-\Delta_{2}+\Delta_{3} & -2 \Delta_{3} & 0 & 0 \\
 -2 \Delta_{3} & -\Delta_{1}+\Delta_{2}+\Delta_{3} & -2 \Delta_{2} & 0 \\
 0 & -2 \Delta_{2} & \Delta_{1}+\Delta_{2}-\Delta_{3} & -2 \Delta_{1} \\
 0 & 0 & -2 \Delta_{1} &  \Delta_{1}-\Delta_{2}-\Delta_{3} \\
\end{array}
\right) . \nonumber
\end{equation}

Substituting the form of $\cH_{\delta}$ in Eq.~\eqref{eq:first_order2} to the first order term in Eq.~\eqref{eq:transAmp_disorder}, and setting  $\overline{J}=1$ throughout, we get, 
\begin{eqnarray}
f^{N}_{r,s}(t,\{\Delta_{k}\})= f_{r,s}^{N}(t) &-& \frac{i}{2} \int_{0}^{t} \sum_{l,k=1}^{N} f^{N}_{r,k}(t)(f^{N}_{k,l}(t'))^{*}f^{N}_{l,s}(t')(\sum_{i=1}^{N-1}u_{i}^{l}\Delta_{i})dt'-\frac{i}{2} \int_{0}^{t} \sum_{l=1}^{N-1}\sum_{k=1}^{N} f^{N}_{r,k}(t)(f^{N}_{k,l}(t'))^{*}f^{N}_{l+1, s}(t')(-2\Delta_{l})dt' \nonumber \\ 
&-& \frac{i}{2} \int_{0}^{t} \sum_{l=1}^{N-1}\sum_{k=1}^{N} f^{N}_{r,k}(t)(f^{N}_{k,l+1}(t'))^{*}f^{N}_{l,s}(t')(-2\Delta_{l})dt' . \nonumber
 \end{eqnarray}
Thus, up to first order in perturbation, $f^{N}_{r,s}( t,\{\Delta_{k}\})$ is simply a linear combination of the random variables $\{\Delta_{k}\}$, of the form,
\begin{equation}
 f^{N}_{r,s}(t,\{\Delta_{k}\}) = f^{N}_{r,s}(t)+ \sum_{i=1}^{N-1} c^{N}_{i}(t) \Delta_{i}, \label{eq:transAmp_final2}
 \end{equation}
where $\{c^{N}_{i}(t)\}$ are complex coefficients given by,
\begin{equation}
c^{N}_{i}(t) = -\frac{i}{2}\sum_{k=1}^{N}f^{N}_{r,k}(t)\left[ \int_{0}^{t}  \sum_{l=1}^{N} u_{i}^{l} (f^{N}_{k,l}(t'))^{*}f^{N}_{l,s}(t')dt' -2 \int_{0}^{t} ( f^{N}_{k,i}(t'))^{*}f^{N}_{i+1, s}(t')dt' -2 \int_{0}^{t}(f^{N}_{k,i+1}(t'))^{*}f^{N}_{i,s}(t')dt' \right]. \label{eq:c-coeff}
\end{equation}

We first note that in the limit of large $N$, the distribution of $f^{N}_{r,s}(t)$ tends towards a normal distribution. This is a direct consequence of the central limit theorem, since $\{\Delta_{i}\}$ are i.i.d random variables. In what follows, we will obtain the exact form of the distribution of $f^{N}_{r,s}(t,\{ \Delta_{k} \})$, specifically, the real and imaginary parts of $f^{N}_{r,s}(t,\{ \Delta_{k} \})$ in terms of $N, t$ and $\delta$. 

Since the $\{\Delta_{i}\}$ are randomly drawn from a uniform distribution between $\left [ -\delta,\delta \right ]$, the joint probability density $P\left(\Delta_{1},\Delta_{2}, \ldots, \Delta_{N} \right)$ is given by,
\begin{equation}
 P \left(\, \Delta_{1},\Delta_{2}, \ldots, \Delta_{N-1} \,\right) = \left\lbrace \begin{array}{cc}
 \frac{1}{(2\delta)^{N-1}}, &  -\delta \leq \Delta_{i} \leq \delta, \; \forall i=1,2,\ldots, N-1. \\
 0, & {\rm otherwise}. 
  \end{array} \right.
\end{equation}
Let $x \equiv \texttt{Re}[f^{N}_{r,s}(t,\{\Delta_{k}\})]$ and $y \equiv \texttt{Im}[f^{N}_{r,s}(t,\{\Delta_{k}\})]$ denote the real and imaginary parts of the transition amplitude in Eq.~\eqref{eq:transAmp_final2}. Then, we may obtain the distribution of $x$ and $y$ as follows:
\begin{eqnarray}
\cP^{\delta,t,N}(x) &=& \int_{\Delta_{1}=-\delta}^{\delta}\ldots\int_{\Delta_{N-1}=-\delta}^{\delta} \left(\prod_{i=1}^{N-1}d\Delta_{i}\right) P(\Delta_{1},\Delta_{2},\ldots,\Delta_{N-1}) \, \delta\left( x - ( \, \texttt{Re}[f^{N}_{r,s}(t)] + \sum_{i=1}^{N-1}\texttt{Re}[c^{N}_{i}(t)]\Delta_{i})\right), \nonumber \\
\cP^{\delta,t,N}(y) &=& \int_{\Delta_{1}=-\delta}^{\delta}\ldots\int_{\Delta_{N-1}=-\delta}^{\delta} \left(\prod_{i=1}^{N-1}d\Delta_{i}\right) P(\Delta_{1},\Delta_{2},\ldots,\Delta_{N-1}) \, \delta\left( y - (\, \texttt{Im}[f^{N}_{r,s}(t)] + \sum_{i=1}^{N-1}\texttt{Im}[c^{N}_{i}(t)] \Delta_{i}) \right). \nonumber 
\end{eqnarray}
Replacing the Dirac delta functions with their Fourier transforms, and then integrating out the $\{\Delta_{k}\}$ variables, we get,
\begin{eqnarray}
\cP^{\delta,t,N}(x) &=& \frac{1}{\sqrt{2\pi}(2\delta)^{N-1}}\int_{\Delta_{1}=-\delta}^{\delta}\ldots\int_{\Delta_{N-1}= -\delta}^{\delta} \int_{k=-\infty}^{\infty} \prod_{i=1}^{N-1}d\Delta_{i} dk \exp{\left(-i k \left( x - \left[ \, \texttt{Re}[f^{N}_{r,s}(t)]+\sum_{i=1}^{N-1}\texttt{Re}[c^{N}_{i}(t)]\Delta_{i}\right]\right)\right)}  \nonumber \\
 &=& \frac{1}{\sqrt{2\pi}(2\delta)^{N-1}} \int_{k=-\infty}^{\infty} dk \exp{\left(-i k ( \, x - \texttt{Re}[f^{N}_{r,s}(t)] \,) \right)} \prod_{i=1}^{N-1} \frac{2\sin \left( k\delta\,\texttt{Re}[c^{N}_{i}(t)] \right)}{k\,\texttt{Re}[c^{N}_{i}(t)]}    \label{eq:dist_real1} \\
&=& \frac{1}{\sqrt{2\pi}(2\delta)^{N-1}} \int_{k=-\infty}^{\infty} dk \exp{\left(-i k ( \, x - \texttt{Re}[f^{N}_{r,s}(t)] \,) \right)} \prod_{i=1}^{N-1} \frac{e^{i \left( k\delta\,\texttt{Re}[c^{N}_{i}(t)] \right)}-e^{-i \left( k\delta\,\texttt{Re}[c^{N}_{i}(t)] \right)}}{i k\,\texttt{Re}[c^{N}_{i}(t)]} \nonumber  \\ \nonumber
&=& \frac{1}{\sqrt{2\pi}(2\delta)^{N-1}} \int_{k=-\infty}^{\infty} dk \exp{\left(-i k ( \, x - \texttt{Re}[f^{N}_{r,s}(t)] \,) \right)}  \frac{\sum_{j=1}^{2^{N-1}} (-1)^{\alpha_{j}}e^{i (k\,\delta\sum_{i=1}^{N-1} (-1)^{r_{i}^{j}}\texttt{Re}[c^{N}_{i}(t)])}}{(i k)^{N-1}\,\prod_{i=1}^{N-1}\texttt{Re}[c^{N}_{i}(t)]}, \nonumber  
\end{eqnarray}
where, $\alpha_{j}, r_{i}^{j} \in [0,1], \; \forall i,j$. Simplifying further, we get,
\begin{eqnarray}
\cP^{\delta,t,N}(x) &=& \frac{1}{\sqrt{2\pi}(2\delta)^{N-1}\prod_{i=1}^{N-1}\texttt{Re}[c^{N}_{i}(t)]} \int_{k=-\infty}^{\infty} dk \frac{\sum_{j=1}^{2^{N-1}} (-1)^{\alpha_{j}}\exp{\left(-i k ( \, x - \texttt{Re}[f^{N}_{r,s}(t)] + \delta \sum_{i=1}^{N-1}(-1)^{r^{j}_{i}}\texttt{Re}[c^{N}_{i}(t)] ) \right)}}{{(ik)}^{N-1}} \nonumber \\ 
& =& \left(\frac{1}{(2\delta)^{N-1}}\right)\left(\frac{1}{\prod_{i=1}^{N-1}\texttt{Re}[c^{N}_{i}(t)]}\right) \sum_{j=1}^{2^{N-1}} (-1)^{u_{j}}(q_{j})^{N-2}\,{\rm Sign}[q_{j}], \label{eq:dist_real2}
\end{eqnarray}
where $u_{j} \in [0,1]$, and $q_{j} (x,\texttt{Re}[f^{N}_{r,s}(t)], \{\texttt{Re}[c^{N}_{i}(t)]\})$ are linear combinations of the form,
\begin{equation}
q_{j} \equiv x - \texttt{Re}[f^N_{r,s}(t)] + \delta\sum_{i=1}^{N-1} (-1)^{r_{i}^{j}}\texttt{Re}[c^{N}_{i}(t)]  , \; r_{i}^{j}\in [0,1], \; \forall i=1,\ldots, N-1. 
\end{equation} 
We may evaluate the distribution of the imaginary part of the transition amplitude in a similar fashion, to get,
\begin{equation}
\cP^{\delta,t,N}(y) = \left(\frac{1}{(2\delta)^{N-1}}\right)\left(\frac{1}{\prod_{i=1}^{N-1}\texttt{Im}[c^{N}_{i}(t)]}\right) \sum_{i=1}^{2^{N-1}} (-1)^{u_{j}}(\tilde{q}_{j})^{N-2}\,{\rm Sign}[\tilde{q}_{j}], \label{eq:dist_im2}
\end{equation}
where the $\tilde{q}_{j} (x,\texttt{Im}[f^{N}_{r,s}(t)], \{\texttt{Im}[c^{N}_{i}(t)]\})$ are linear combinations of the form,
\begin{equation}
\tilde{q}_{j} \equiv y - \texttt{Im}[f^{N}_{r,s}(t)] + \delta\sum_{i=1}^{N-1} (-1)^{r_{i}^{j}}\texttt{Im}[c^{N}_{i}(t)]  , \; r_{i}^{j}\in [0,1], \; \forall i=1,\ldots, N-1. 
\end{equation}
We see from Eq.~\eqref{eq:dist_real1} that the limiting distribution in the case of no disorder ($\delta \rightarrow 0$), is indeed a delta distribution peaked around $\texttt{Re}[f^{N}_{r,s}(t)]$:
\begin{equation}
\lim_{\delta\rightarrow 0}\cP^{\delta,t,N}(x) = \frac{1}{\sqrt{2\pi}}\int_{k=-\infty}^{\infty} dk \exp{\left(-i k ( \, x - \texttt{Re}[f^{N}_{r,s}(t)] \,) \right)} = \delta\left( x - \texttt{Re}[f^{N}_{r,s}(t)]\right).
\end{equation}

\begin{figure}  [H]
\centering
\begin{subfigure}{
\includegraphics[width=0.47\textwidth]{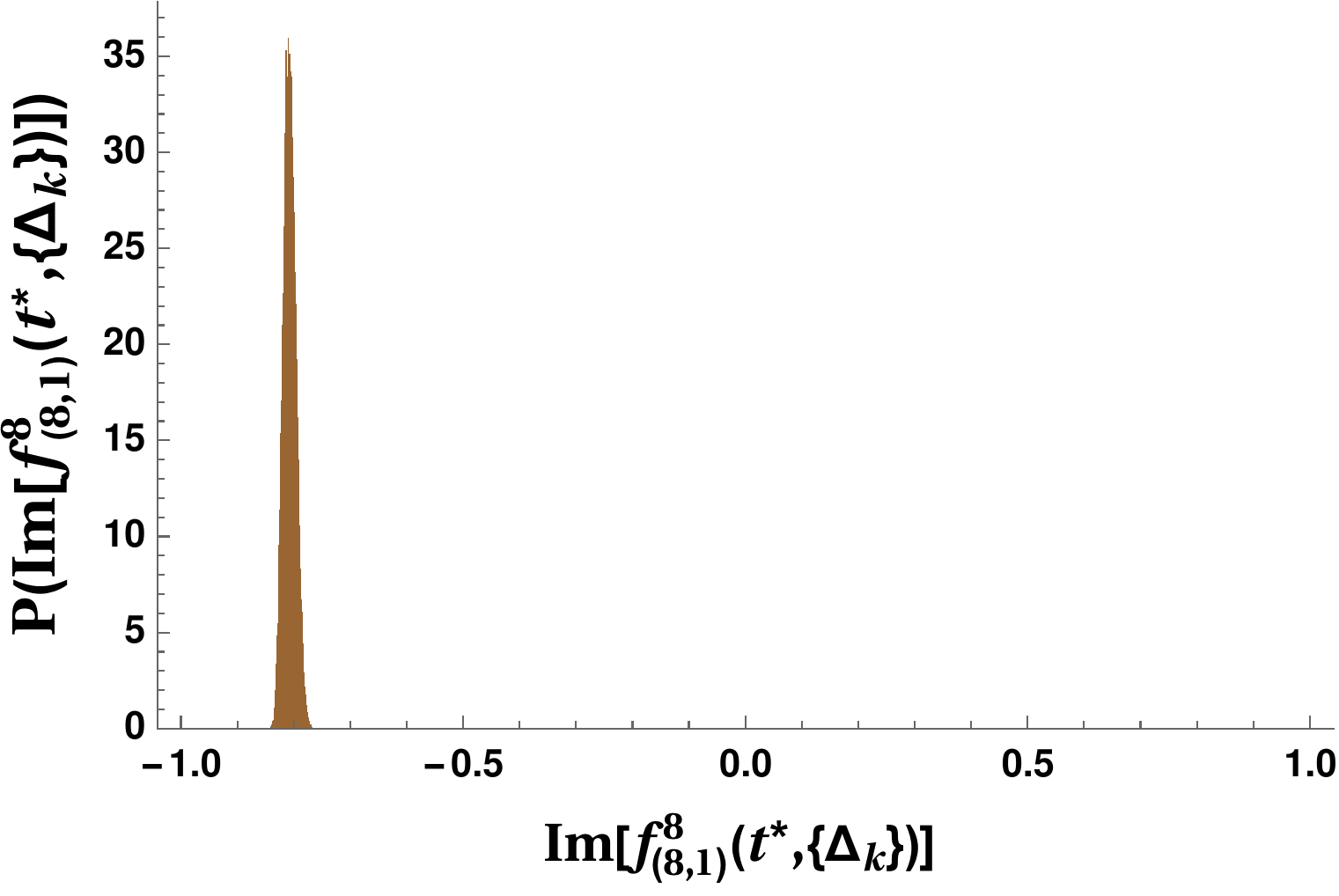}
}
\end{subfigure}
\begin{subfigure}{
\includegraphics[width=0.47\textwidth]{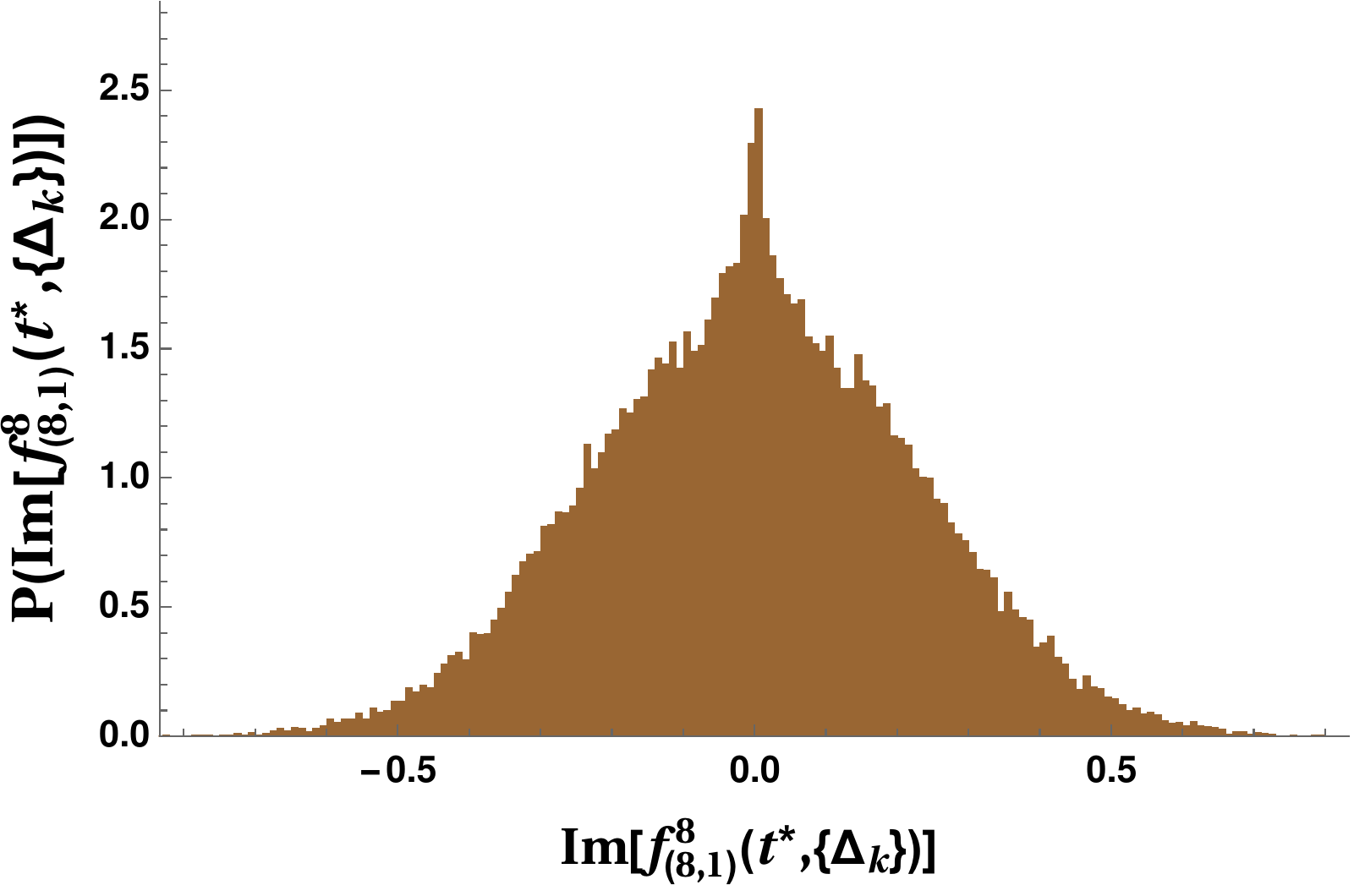}
}
\end{subfigure}

\caption{Distribution of $\texttt{Im}[f_{8,1}^{8}(t^{*}, \{\Delta_{k}\})]$ 
over different disorder realizations drawn from a uniform distribution with disorder strengths $\delta = 0.001$ and $\delta=1$, respectively.}
\label{fig:f_dist_Im}
\end{figure}




Finally, we compute the disorder-averaged value of the transition amplitude 
upto $O(\cH_{\delta}^{2})$. We first modify the expression in Eq.~\eqref{eq:transAmp_final2} to include the second-order perturbation terms:
\begin{equation}
 f^{N}_{r,s}(t,\{\Delta_{k}\}) = f_{r,s}^{N}(t)+ \sum_{i=1}^{N-1}c^{N}_{i}(t) \Delta_{i} + \sum_{i,j=1}^{N-1}d^{N}_{ij} \Delta_{i}\Delta_{j} + \ldots , \label{eq:transAmp_final3}
 \end{equation}
where $\{d^{N}_{ij}\}$ are complex coefficients which are convolutions of the zero-disorder transition amplitude, similar to $\{c^{N}_{i}(t)\}$. Next, using the fact that the random couplings $\{\Delta_{i}\}$ are drawn from a uniform distribution, we obtain,
\begin{eqnarray}\label{eq:trans_ampAvg}
\langle f^{N}_{r,s}(t,\{\Delta_{k}\}) \rangle_{\delta}&=&\frac{1}{(2\delta)^{N-1}}\int_{-\delta}^{\delta} \left( f^{N}_{r,s}(t) + \sum_{i=1}^{N-1}c^{N}_{i}(t)\Delta_{i} +\sum_{l,m=1}^{N-1} d^{N}_{lm}(t) \Delta_{l}\Delta_{m}+\ldots \right) \prod_{i=1}^{N-1}d\Delta_{i} \nonumber \\ 
&=& f^{N}_{r,s}(t)+ \frac{\delta^{2}}{3}\sum_{i}d^{N}_{ii}(t) + O(\delta^{4}) .
\end{eqnarray}
The second moment of $f^{N}_{r,s}(t,\{\Delta_{k}\})$,
\begin{eqnarray}\label{eq:2ndmoment}
\left\langle \, (f^{N}_{r,s}(t,\{\Delta_{k}\}))^{2} \, \right\rangle_{\delta}&=&\frac{1}{(2\delta)^{N-1}} \int_{-\delta}^{\delta} \left( f^{N}_{r,s}(t) + \sum_{i=1}^{N-1}c^{N}_{i}(t)\Delta_{i} +\sum_{l,m=1}^{N-1} d^{N}_{lm}(t) \Delta_{l}\Delta_{m})+\ldots \right)^{2} \prod_{i=1}^{N-1}d\Delta_{i} \nonumber \\ 
&=& (f^{N}_{r,s}(t))^{2} + \frac{\delta^{2}}{3}\left( 2 f^{N}_{r,s}(t) \sum_{l=1}^{N-1}d^{N}_{ll}(t)+\sum_{j=1}^{N-1}(c^{N}_{j}(t))^{2}\right) \nonumber \\
&&  + \frac{\delta^{4}}{5}\sum_{l=1}^{N-1}(d^{N}_{ll}(t))^{2}+ \frac{\delta^{4}}{9}\sum_{l \neq m=1}^{N-1}(d^{N}_{lm}(t))^{2} + O(\delta^{6}).
\end{eqnarray}
We can now calculate the variance from Eq~\ref{eq:trans_ampAvg} and Eq~\ref{eq:2ndmoment} as follows:
 \begin{eqnarray}\label{eq:variance}
{\rm Var}[f^{N}_{r,s}(t,\{\Delta_{k}\})] &=& \langle \, (f^{N}_{r,s}(t,\{\Delta_{k}\}))^{2} \, \rangle_{\delta} - \langle f^{N}_{r,s}(t,\{\Delta_{k}\}) \rangle^{2}_{\delta} \nonumber \\
&=& \frac{\delta^{2}}{3}\sum_{j=1}^{N-1}(c^{N}_{j}(t))^{2}+\delta^{4}\left(\frac{1}{5}\sum_{l=1}^{N-1}(d^{N}_{ll}(t))^{2}+ \frac{1}{9}\sum_{l \neq m=1}^{N-1}(d^{N}_{lm}(t))^{2}-\frac{1}{9}\left(\sum_{l=1}^{N-1}d^{N}_{ll}(t)\right)^{2} \right)  \nonumber \\ &&+ O(\delta^{6}). 
 \end{eqnarray}
To summarize, from Eq.~\eqref{eq:trans_ampAvg} we see that as  $\delta \rightarrow 0$, $\langle f^{N}_{r,s}(t,\{\Delta_{k}\}) \rangle_{\delta}$  approaches the zero-disorder value $f^{N}_{r,s}(t)$. As expected, the variance given in Eq.~\eqref{eq:variance} vanishes in this limit. However as the disorder strength $\delta$ increases, $\langle f^{N}_{r,s}(t,\{\Delta_{k}\}) \rangle_{\delta}$ deviates from the no-disorder case, and the variance also starts growing since terms of $O(\delta^{2})$ become increasingly significant now. 
 
\end{widetext}

%
%
%
%
%
%

\end{document}